\documentclass[12pt,draftcls,onecolumn,peerreview]{IEEEtran}

\usepackage{setspace}
\newtheorem{prop}{Proposition}

\newtheorem{cor}{Corollary}

\newtheorem{lm}{Lemma}

\newtheorem{thm}{Theorem}

\newcommand{\bal}{\begin{align}}
\newcommand{\eal}{\end{align}}
\newcommand{\be}{\begin{eqnarray}}
\newcommand{\ee}{\end{eqnarray}}
\newcommand{\benn}{\begin{eqnarray*}}
\newcommand{\eenn}{\end{eqnarray*}}
\def\IR{\rm I \kern-0.20em R}
\newcommand{\utwi}[1]{\mbox{\boldmath $ #1$}}

\newcommand{\bthm}{\begin{thm}}
\newcommand{\ethm}{\end{thm}}

\newcommand{\bcor}{\begin{cor}}
\newcommand{\ecor}{\end{cor}}
\newcommand{\bprop}{\begin{prop}}
\newcommand{\eprop}{\end{prop}}
\newcommand{\blm}{\begin{lm}}
\newcommand{\elm}{\end{lm}}
\newcommand{\beq}{\begin{equation}}
\newcommand{\eeq}{\end{equation}}
\newcommand{\ber}{\begin{eqnarray}}
\newcommand{\eer}{\end{eqnarray}}
\newcommand{\bseq}{\begin{subequations}}
\newcommand{\eseq}{\end{subequations}}

\newcommand{\bproof}{\begin{proof}}
\newcommand{\eproof}{\end{proof}}

\renewcommand{\i}{\imath}
\renewcommand{\j}{\jmath}

\newcommand{\gm}{\gamma}

\newcommand{\sm}{\sigma}
\newcommand{\Sm}{\mathnormal\Sigma}
\newcommand{\al}{\alpha}
\newcommand{\om}{\omega}

\newcommand{\diag}{\mathop{\mbox{\rm diag}}}

\newcommand{\tr}{\mathop{\mbox{\rm Tr}}}

\newcommand{\bit}{\begin{itemize}}
\newcommand{\eit}{\end{itemize}}
\newcommand{\ben}{\begin{enumerate}}
\newcommand{\een}{\end{enumerate}}
\newcommand{\bdesc}{\begin{description}}
\newcommand{\edesc}{\end{description}}
\newcommand{\beqarrn}{\begin{eqnarray*}}
\newcommand{\eeqarrn}{\end{eqnarray*}}
\newcommand{\bproofof}{\begin{proofof}}
\newcommand{\eproofof}{\end{proofof}}
\newenvironment{rem}{\begin{trivlist}\item[]{\bf
Remark:}\hspace{4mm}}{\end{trivlist}}
\newcommand{\brem}{\begin{rem}}
\newcommand{\erem}{\end{rem}}
\newenvironment{rems}{\begin{trivlist}\item[]{\bf
Remarks}\begin{itemize}}{\end{itemize}\end{trivlist}}
\newcommand{\brems}{\begin{rems}}
\newcommand{\erems}{\end{rems}}
\newtheorem{fact}{Fact}
\newcommand{\bfact}{\begin{fact}}
\newcommand{\efact}{\end{fact}}
\newtheorem{examp}{Example}
\newcommand{\bexamp}{\begin{examp}\rm}
\newcommand{\eexamp}{\end{examp}}
\newtheorem{defn}{Definition}
\newcommand{\bdefn}{\begin{defn}\rm}
\newcommand{\edefn}{\end{defn}}

\newtheorem{alg}{Algorithm}
\newcommand{\balg}{\begin{alg}}
\newcommand{\ealg}{\end{alg}}

\newtheorem{prob}{Problem}
\newcommand{\bprob}{\begin{prob}}
\newcommand{\eprob}{\end{prob}}

\newcommand{\bvtm}{\begin{verbatim}}
\newcommand{\bfig}{\begin{figure}}
\newcommand{\efig}{\end{figure}}
\newcommand{\bcen}{\begin{center}}
\newcommand{\ecen}{\end{center}}

\def\Bbb{\mathbb }

\def\eptn {{\Bbb E}}

\long\def\comment#1{}

\def\lb{\lambda}

\def \n2{{N_0 \over 2}}

\def \h5{\hspace{0.5in}}

\newcommand{\bee}{{\boldsymbol{e}}}
\newcommand{\bff}{{\boldsymbol{f}}}
\newcommand{\bg}{{\boldsymbol{g}}}
\newcommand{\bh}{{\boldsymbol{h}}}

\newcommand{\bn}{{\boldsymbol{n}}}

\newcommand{\bp}{{\boldsymbol{p}}}

\newcommand{\br}{{\boldsymbol{r}}}

\newcommand{\bu}{{\boldsymbol{u}}}
\newcommand{\bv}{{\boldsymbol{v}}}
\newcommand{\bw}{{\boldsymbol{w}}}
\newcommand{\bx}{{\boldsymbol{x}}}
\newcommand{\by}{{\boldsymbol{y}}}
\newcommand{\bz}{{\boldsymbol{z}}}

\newcommand{\bA}{{\boldsymbol{A}}}
\newcommand{\bB}{{\boldsymbol{B}}}
\newcommand{\bC}{{\boldsymbol{C}}}
\newcommand{\bD}{{\boldsymbol{D}}}

\newcommand{\bF}{{\boldsymbol{F}}}
\newcommand{\bG}{{\boldsymbol{G}}}
\newcommand{\bH}{{\boldsymbol{H}}}
\newcommand{\bI}{{\boldsymbol{I}}}

\newcommand{\bP}{{\boldsymbol{P}}}
\newcommand{\bQ}{{\boldsymbol{Q}}}

\newcommand{\bS}{{\boldsymbol{S}}}
\newcommand{\bT}{{\boldsymbol{T}}}

\newcommand{\bW}{{\boldsymbol{W}}}
\newcommand{\bX}{{\boldsymbol{X}}}

\newcommand{\bSm}{\boldsymbol{\Sigma}}
\newcommand{\bDt}{\boldsymbol{\Delta}}

\newcommand{\bOm}{\boldsymbol{\Omega}}

\newcommand{\bPsi}{{\utwi{\mathnormal\Psi}}}

\usepackage{tikz}
\usetikzlibrary{arrows}
\tikzstyle{block}=[draw opacity=0.7,line width=1.4cm]

\usepackage{amsthm}

\newtheorem{Theorem}{Theorem}

\newtheorem{Proposition}{Proposition}

\theoremstyle{remark}

\usepackage{latexsym,amssymb,amsmath,dsfont}
\usepackage{cite,mathtools}

\usepackage{algorithmic,framed}
\usepackage{transparent}

\title{Wireless MIMO Switching: Weighted Sum Mean Square Error and Sum Rate Optimization}

\author{Fanggang~Wang, Xiaojun~Yuan, Soung Chang~Liew and Dongning~Guo\\
\thanks{This work was presented in part at the IEEE International
Symposium on Information Theory, Cambridge, MA, July 2012 \cite{wms12isit}.}
}

\begin{document}
\IEEEoverridecommandlockouts

\maketitle
\thispagestyle{headings}
\setcounter{page}{1}
\pagestyle{headings}

\begin{abstract} \label{abs}
This paper addresses  joint transceiver and relay design for a wireless multiple-input-multiple-output (MIMO) switching scheme that enables data exchange among multiple users. Here, a multi-antenna relay linearly precodes the received (uplink) signals from multiple users before forwarding the signal in the downlink, where the purpose of precoding is to let each user receive its desired signal with interference from other users suppressed. The  problem of optimizing the precoder based on various design criteria is typically non-convex and  difficult to solve. The main contribution of this paper is a unified  approach to solve the weighted sum mean square error (MSE) minimization and weighted sum rate maximization problems in MIMO switching. Specifically,  an iterative algorithm is proposed for jointly optimizing the relay's precoder  and the users' receive filters to minimize the weighted sum MSE. It is also shown that the weighted sum rate maximization problem can be reformulated as an iterated weighted sum MSE minimization problem and can therefore be solved similarly to the case of weighted sum MSE minimization. With properly chosen initial values, the proposed iterative algorithms are asymptotically optimal in both high and low signal-to-noise ratio (SNR) regimes for MIMO switching, either with or without self-interference cancellation (a.k.a., physical-layer network coding). Numerical results show that the optimized MIMO switching scheme based on the proposed algorithms significantly outperforms  existing approaches in the literature.

\end{abstract}


\section{Introduction}

Physical-layer network coding (PNC) has received much attention in recent years \cite{Zhang06physical-layernetwork}. The simplest communication model for PNC is a two-way relay channel, in which two users accomplish bidirectional data exchange in two phases of transmission with the help of a relay. Significant progress has been made in approaching the capacity of two-way relay channel (see \cite{Zhang06physical-layernetwork,Bor07,rui09,ding10,meixia11j,yonghui,soung12,xiaojun12twireless,xiaojun13tcom} and the references therein).

Multi-way relaying, in which multiple users exchange data via a single relay, has been studied more recently \cite{ wms12jsac,Joung10,ey10,moh09,Cui08,aut10,guo12jsac,den09x,Cadambe12,Ong11,Ong12}. In \cite{den09x,Cadambe12,Ong11,Ong12}, the relay is equipped with a single antenna. The use of multiple antennas at the relay provides extra spatial degrees of freedom that can boost throughput significantly. A multi-antenna relay that performs one-to-one mapping from the inputs to the outputs (i.e., to switch traffic in a one-to-one manner among the end users) is called a MIMO switch \cite{wms12jsac,moh09}. Various traffic patterns have been studied in MIMO relaying, including pairwise data exchange \cite{moh09,Joung10,ey10}, where the users form pairs and data exchange is within each pair, and full data exchange, where each user broadcasts to all other users \cite{Cui08,aut10,guo12jsac}. Reference \cite{wms12jsac} further generalized pairwise data exchange to arbitrary unicast, in which each user sends data to one other user and could receive data from a different user. Arbitrary unicast is interesting because any traffic pattern, including unicast, multicast, broadcast, or any mixture of them, can be realized by scheduling  a sequence of unicast flows.

Joint transceiver and relay design for MIMO switching with simultaneous unicast has been reported in \cite{aut10,Joung10,ey10,wms12jsac}. Zero-forcing relaying was first proposed in \cite{Joung10} to realize pairwise data exchange. Zero-forcing relay with PNC, which employs self-interference cancellation, can improve system throughput considerably \cite{wms12jsac}. However, zero-forcing involves channel inverse operations that incur significant power penalties when the channel gain matrix is ill-conditioned. To alleviate power penalties, minimum mean square error (MMSE) relaying was proposed in \cite{Joung10} and \cite{aut10}, which achieves better performance for practical signal-to-noise ratios (SNRs). All preceding works aim at suppressing the interference or the MSE.

In this paper, two families of optimization problems, namely, weighted sum MSE minimization and weighted sum rate maximization are  tackled under one unified framework. The sum rate metric is directly related to user experience, and hence as important as the MSE metric. Specifically, to minimize weighted sum MSE, an algorithm can be devised to iteratively optimize the relay's precoder and the users' receive filters. We  show that the weighted sum rate maximization problem is converted to an iterated weighted sum MSE minimization problem, which admits an iterative solution. In the low and high SNR regimes, analytical results are provided on the properties of the asymptotically optimal solutions and the convergence conditions of the proposed algorithms. Numerical results show that the proposed algorithms significantly outperform existing approaches in \cite{aut10,Joung10,wms12jsac}.


The following notational convention is adopted throughout this paper: Scalars are in normal fonts; boldface lower-case letters denote vectors and boldface upper-case letters denote matrices; $\diag\{\bx\}$ denotes a diagonal square matrix whose diagonal consists of  the elements of $\bx$; $\diag\{\bX\}$ denotes a column vector formed by the diagonal elements of $\bX$. $[\bX]_{\text{diag}}$ represents a diagonal matrix with the same diagonal elements as $\bX$. Denote by $\mathcal{C}_{a,b}$ the covariance of two zero-mean random variables $a$ and $b$, i.e., $\mathcal{C}_{a,b}=\eptn[ab^*]$. The operation $vec(\bX)$ is to stack each column of the matrix $\bX$ on top of the right adjacent column; $mat(\cdot)$ is the inverse operator of $vec(\cdot)$; $(\cdot)^\dag$ denotes Moore-Penrose pseudo inverse \cite{roger90}; and $\otimes$ denotes the Kronecker product.

The remainder of the paper is organized as follows: Section II introduces the background of wireless MIMO switching. Weighted sum MSE minimization and weighted sum rate maximization are discussed in Section III and IV, respectively. In Section V, asymptotically optimal solutions  are derived. Section VI presents simulation results. Section VII concludes this paper.

\section{System Description} \label{sec.SystemDes}

The system model is illustrated in Fig.~\ref{diagram}. There are $K$ users, numbered from $1$ to $K$, each equipped with a single antenna. These users communicate via a relay with $N$ antennas and there is no direct link between any two users. Throughout the paper, we focus on the pure unicast case, in which each user transmits to one other user only.
Let $\pi(\cdot)$ specify a switching pattern, which can be represented as follows: user $i$ transmits to $j=\pi(i)$ for every $i\in\{1,\cdots,K\}$. The pure unicast switching pattern can be equivalently represented by a permutation matrix $\bP$.
Let $\bee_j$ denote the $j$th column of an identity matrix. Then the $i$th column of $\bP$ is equal to $\bee_j$ if $\pi(i)=j$, i.e., $\bp_i=\bee_{\pi(i)}=\bee_j$. If the diagonal elements of permutation $\bP$ are all zero, it is also called a {\em derangement}. In particular, a symmetric derangement ($\bP=\bP^T$) realizes a pairwise data exchange.
In general, any traffic flow pattern among the users can be realized by scheduling a set of different unicast traffic flows \cite{wms12jsac}.

\tikzstyle{place}=[circle,draw=blue!50,fill=blue!20,thick,
inner sep=0pt,minimum size=6mm]
\tikzstyle{transition}=[rectangle,draw=black!100,fill=none,semithick,
inner sep=0pt,minimum size=6mm]
\tikzstyle{mimoswitch}=[rectangle,draw=black!100,fill=none,semithick,
inner sep=0pt,minimum width=12mm, minimum height=26mm]
\tikzstyle{relay}=[rectangle,draw=black!100,fill=none,semithick,
inner sep=0pt,minimum width=25mm, minimum height=10mm]

\def\antenna{%
    -- +(0mm,4.0mm) -- +(2.625mm,7.5mm) -- +(-2.625mm,7.5mm) -- +(0mm,4.0mm)
}

\begin{figure}
\centering
\begin{tikzpicture}[node distance=10em]
\node (relay) [relay] {Relay};
\node (1) [above right of=relay] [transition] {1};
\node (2) [above left of=relay] [transition] {2};
\node (3) [below left of=relay] [transition] {3};
\node (N) [below right of=relay] [transition] {$K$};

\draw[color=black,semithick] (1.north) \antenna;
\draw[color=black,semithick] (2.north) \antenna;
\draw[color=black,semithick] (3.north) \antenna;
\draw[color=black,semithick] (N.north) \antenna;
\path (relay.north west) to node (a) [pos=.2,inner sep=0] {} (relay.north) to node (b) [pos=.8,inner sep=0] {} (relay.north east);
\draw[color=black,semithick] (a) \antenna  (b) \antenna;

\draw [dashed] (-.8,.8) to node [above]  {$N$} (.8,.8);
\draw [color=blue, dashed, semithick, ->] (1.3,1.5) -- (1.9,2.1); \draw [semithick, ->] (2.1,1.9) -- (1.5,1.3);
\draw [color=blue, dashed, semithick, ->] (-1.3,1.5) -- (-1.9,2.1); \draw [semithick, ->] (-2.1,1.9) -- (-1.5,1.3);
\draw [color=blue, dashed, semithick, ->] (1.3,-0.8) -- (1.9,-1.4); \draw [semithick, ->] (2.1,-1.2) -- (1.5,-0.6);
\draw [color=blue, dashed, semithick, ->] (-1.3,-0.8) -- (-1.9,-1.4); \draw [semithick, ->] (-2.1,-1.2) -- (-1.5,-0.6);

\draw [color=blue, dashed, semithick, ->] (-2.8,-4.2) -- (-1.7,-4.2); \draw [semithick, ->] (-2.8,-3.7) -- (-1.7,-3.7);
\node at (0.2,-3.7) {\small{Uplink phase}};
\node at (0.4,-4.2) {\small{Downlink phase}};

\draw [dashed] (-.5,-1.5) to node [auto] {} (.5,-1.5);
\end{tikzpicture}
\caption{Wireless MIMO switching.}
\label{diagram}
\end{figure}
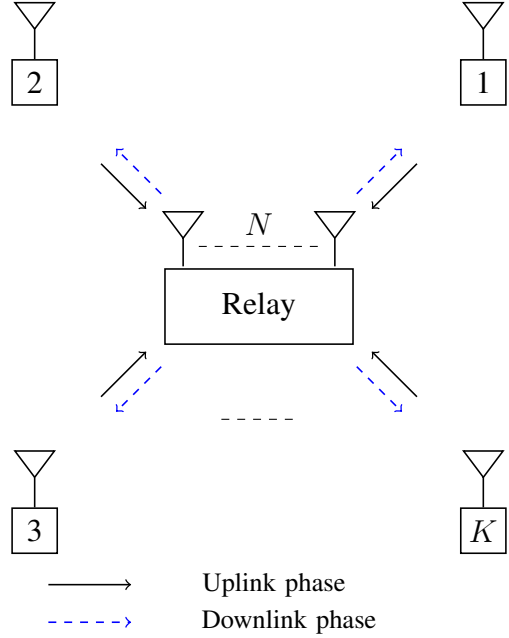

Each round of data exchange consists of one uplink phase and one downlink phase. The uplink phase sees simultaneous transmissions from the users to the relay; the downlink phase sees one transmission from the relay to the users. We assume that the two phases are of equal duration.

In the uplink phase, let $\bx = [x_1, \cdots, x_K]^T$ be the vector representing the signals transmitted by the users. Let $\by=[y_1, \cdots, y_N]^T$ be the received signals at the relay's antennas, and $\bu=[u_1, \cdots, u_N]^T$  be the noise vector with independent and identically distributed (i.i.d.) samples following the circularly-symmetric complex Gaussian (CSCG) distribution, denoted by $\mathcal{CN}(0, \gamma^2)$, where $\gm^2$ is the noise variance at the relay. Then
\begin{equation} \label{formula_uplink}
\by=\bH \bx +\bu,
\end{equation}
where $\bH\in\mathbb{C}^{N\times K}$  is the uplink channel matrix. We assume that all uplink signals are independent Gaussian with zero mean and  their powers are constrained as 
\be
\mathbb{E}\{|x_i|^2\} =q_i \leq Q_i,\quad i=1,\cdots,K. 
\ee 

Upon receiving $\by$, the relay precodes $\by$ with matrix $\bG$ and forward $\bG\by$ in the downlink phase, where the transmit power of the relay is upper-bounded by $P_r$, i.e.
\be
\tr\left[\bG\left(\bH\bQ\bH^H+\gm^2\bI\right)\bG^H\right]\leq P_r,
\ee
where $\bQ = \diag\{q_1,\cdots,q_K\}$.

The signals received by all users are collectively represented in the vector form as
\begin{equation} \label{r}
 \br = \bF \bG\by + \bw = \bF \bG\bH \bx + \bF \bG\bu  + \bv,
\end{equation}
where $\bF$ is the downlink channel matrix. Note that we assume the channel matrices $\bH$ and $\bF$ are either full row rank or column rank, whichever is less. $\bv$ is the noise vector at the receivers, with i.i.d.  samples following the CSCG distribution, i.e., $ v_k \sim \mathcal{CN}(0, \sigma^2)$, where $\sm^2$ is the noise variance. The signal $r_{\pi(i)}$ received by user $\pi(i)$ is used to recover the message from user $i$. Each user scales its received signal before estimation. We use a diagonal matrix $\bC$ to denote the combination of all the scaling factors. It is not difficult to see that interference-free switching can be achieved in the absence of noise by designing $\bG$ to ensure $\bP^T\bC\bF\bG\bH$ is diagonal. In the presence of noise, $\bG$ and $\bC$ are chosen to optimize certain performance metric.
The MIMO precoder $\bG$ shall be chosen to manage interference based on the knowledge of the uplink and downlink channels $\bH$ and $\bF$. Furthermore, physical-layer network coding technique can be used if the users can cancel self-interference in the received signal. We refer to the precoding schemes without network coding as \emph{non-PNC schemes}, and subsequent ones with network coding as \emph{PNC schemes}.


\section{Weighted Sum MSE Minimization}

\subsection{Problem Formulation}

We start with the weighted sum MSE minimization problem. Each user $j$ aims to estimate the signal $x_i$ from user $i$ based on the received signal $r_j$, where $j = \pi(i)$. For Gaussian signaling and the linear system in \eqref{r}, the MMSE estimator of $x_i$ is given by $\hat x_i = c_j r_j$, and the resulting MSE is given by $|x_i - \hat x_i|^2$, where $c_j$ is a scalar coefficient. Then, the overall weighted sum MSE can be compactly written as $\|\bW^{\frac12} (\bP\bx-\bC\br)\|^2$, where $\bW = \diag\{w_1,\cdots,w_K\}$, $\bC = \diag\{c_1,\cdots,c_K\}$, and each $w_j$ represents the MSE weighting coefficient of receiver $j$. The above discussion is under the assumption that the network-coding technique (or self-interference cancellation) is not applied at the receivers.

In the case where network coding is allowed, the goal is to make $\hat x_i$ close to $x_i$ after removing self-interference. The decision statistics for all $K$ users can be expressed in vector form as $\bC\br-\bB\bx$, where $\bB$ is a diagonal matrix consisting of the weighting factors of self-interference. The weighted sum square error is then given by $\|\bW^{\frac{1}{2}}\left(\bP\bx-(\bC\br-\bB\bx)\right)\|^2$, where $\|\cdot\|$ represents the Euclidean norm, and $\bW$ is a diagonal matrix with the $i$th diagonal element being the weight of the MSE at user $i$.

We first assume that each user transmit with its maximum power, i.e., $q_i=Q_i$, $i=1,\cdots,K$. The case of user power control, i.e., $q_i\leq Q_i$, $i=1,\cdots,K$, will be discussed later in Section III.E.  Based on the above, the weighted sum MSE minimization problem is formulated as
\begin{subequations} \label{opt_mse}
\begin{align}
\mathop\text{minimize}_{\boldsymbol{G},\boldsymbol{B},\boldsymbol{C}}  &\qquad \mathbb{E}\left\|\bW^{\frac{1}{2}}\left((\bP+\bB) \bx - \bC \br\right)\right\|^2 \label{opt_mse_a}\\
\text{subject to} &\qquad  \text{Tr}\left[\bG\left(\bH\bQ\bH^H+\gm^2\bI\right)\bG^H\right]\le P_r \label{opt_mse_b}\\
& \qquad \bB\text{ and } \bC \text{ are diagonal}.
\end{align}
\end{subequations}
The case of no network coding admits the same formulation except that the matrix $\bB$ is set to ${\bf0}$. Problem \eqref{opt_mse} involves the joint optimization of $\bG$, $\bB$ and $\bC$.
Unfortunately, this problem is non-convex, and thus is in general difficult to solve. We next propose an iterative algorithm to find a suboptimal solution to \eqref{opt_mse}, as described below.

Following the convention in \cite{Joham05,Joung10}, we define
\be \label{al_mmse}
\bar\bG\triangleq\al^{-1}\bG\quad \text{and}\quad \bar\bC\triangleq\al\bC,
\ee
where $\al$ is a scaling factor. Later, we will see that $\al$ is introduced to meet the power constraint at the relay. Then the weighted sum MSE \eqref{opt_mse_a} is expanded as
\bal \label{opt_mse_x1a}
\mathcal{J}(\bar\bG, \al, \bB, \bar\bC)\triangleq &\tr\bigg[\bW\bigg(\bP\bQ\bP^T+\bB\bQ\bB^H -2\Re\{\bar\bC\bF\bar\bG\bH\bQ(\bP+\bB)^H\} \notag\\ &+\bar\bC\bF\bar\bG\left(\bH\bQ\bH^H+\gm^2\bI\right)\bar\bG^H\bF^H\bar\bC^H +\sm^2\al^{-2}\bar\bC \bar\bC^H\bigg)\bigg].
\end{align}
We notice that, for fixed $(\bB,\bar\bC)$, this optimization problem can be solved by the KKT conditions similarly to \cite{Joham05,Joung10}. On the other hand, for fixed $(\bar\bG,\al)$, problem \eqref{opt_mse} also reduces to a quadratic program w.r.t. $(\bB,\bar\bC)$. In fact, both optimization problems admit explicit analytical solution. Thus, we can iteratively optimize $(\bar\bG,\al)$ and $(\bB,\bar\bC)$, yielding an approximate solution to problem \eqref{opt_mse}.

It is worth noting that setting $\bB={\bf 0}$  and $\bar\bC=\bI$ reduces problem \eqref{opt_mse} to the conventional MMSE relaying design problem (without PNC) studied in \cite{aut10,Joung10}. We will show later that the flexibility of choosing $\bB$ and $\bar\bC$  provides significant performance gains.

\subsection{Optimal $(\bar\bG,\al)$ for Fixed $(\bB,\bar \bC)$}
For fixed $(\bB,\bar\bC)$, problem \eqref{opt_mse} reduces to
\begin{subequations} \label{opt_mse_1}
\begin{align}
\mathop\text{minimize}_{\boldsymbol{\bar G},\al}  &\qquad \mathcal{J}(\bar\bG, \al) \label{opt_mse_1a}\\
\text{subject to} &\qquad  \al^2\text{Tr}\left[\bar\bG\left(\bH\bQ\bH^H+\gm^2\bI\right)\bar\bG^H\right]\le P_r. \label{opt_mse_1b}
\end{align}
\end{subequations}
In the above, $\mathcal{J}(\bar\bG, \al)$ is a shorthand of $\mathcal{J}(\bar\bG, \al,\bB,\bar\bC)$ for fixed $(\bB, \bar\bC)$. Similar notation will be used throughout without further notice. We use the Lagrangian method to solve problem \eqref{opt_mse_1}. The Lagrangian function is written as
\bal \label{Jmmse}
\mathcal{L}&(\bar\bG,\al,\lambda)= \notag \\&\tr\left\{\bW\left[\bP\bQ\bP^T+\bB\bQ\bB^H-2\Re\{\bar\bC\bF\bar\bG\bH\bQ(\bP+\bB)^H\}+ \bar\bC\bF\bar\bG\bH\bQ\bH^H\bar\bG^H\bF^H\bar\bC^H\right.\right. \notag\\
&+\left.\left. \gm^2\bar\bC\bF\bar\bG\bar\bG^H\bF^H\bar\bC^H +\sm^2\al^{-2}\bar\bC\bar\bC^H\right]\right\} +\lambda \left(\text{Tr}[\al^2\bar\bG\left(\bH\bQ\bH^H+\gm^2\bI\right)\bar\bG^H]- P_r\right),
\end{align}
where $\lambda$ is the Lagrangian multiplier. By setting $\frac{\partial \mathcal{L}}{\partial \lb}=\frac{\partial \mathcal{L}}{\partial \al}=0$, we obtain
\be
\frac{\partial \mathcal{L}}{\partial \lb}=0 &\Rightarrow&  \al^{2} = {\frac{P_r}{\text{Tr}\left[\bar\bG^{} \left(\bH\bQ\bH^H+\gm^2\bI\right)\bar\bG^H\right]}}, \label{lbx}\\
\frac{\partial \mathcal{L}}{\partial \al}=0 &\Rightarrow& \al^4 = \frac{\sm^2 \tr[\bW\bar\bC\bar\bC^H]}{\lb^{}{\text{Tr}\left[\bar\bG^{} \left(\bH\bQ\bH^H+\gm^2\bI\right)\bar\bG^H\right]}}.\label{alx}
\ee
Hence,
\be \label{68}
\lb\al^2=\frac{\sm^2}{P_r}\tr[\bW\bar\bC\bar\bC^H].
\ee
By setting $\frac{\partial \mathcal{L}}{\partial \bar\bG}=0$ and using \eqref{68}, we obtain the optimal precoder $\bG^{opt}$. Hence the following proposition:

\begin{Proposition}  \label{prop_mmse_1}
For fixed $(\bB,\bar\bC)$, the optimal precoder to problem \eqref{opt_mse_1} is  $\bG^{opt}=\al^{}\bar\bG^{opt}$ with
\begin{subequations}  \label{10}
\bal
\bar\bG^{opt} &= \left(\frac{\sm^2}{P_r} \tr[\bW\bar\bC\bar\bC^H]\bI + \bF^H\bar\bC^H\bW\bar\bC\bF\right)^{-1}\bF^H \bar\bC^H \bW(\bP+\bB) \bQ\bH^H \left(\bH\bQ\bH^H+\gm^2\bI\right)^{-1} \label{tGopt_mmse}\\
\al^{opt} &= P_r^{\frac{1}{2}}\left(\tr\left[\bar\bG^{opt} \left(\bH\bQ\bH^H+\gm^2\bI\right)(\bar\bG^{opt})^H\right]\right)^{-\frac{1}{2}}. \label{al}
\end{align}
\end{subequations}
\end{Proposition}
As mentioned earlier, $\al$ in \eqref{al} is just a scaling factor to meet the relay power constraint.

\subsection{Optimal $(\bB,\bar\bC)$ for Fixed $(\bar\bG,\al)$}

For fixed $(\bar\bG,\al)$, the optimization problem \eqref{opt_mse} reduces to an unconstrained problem as follows:
\begin{align}
\mathop\text{minimize}_{\boldsymbol{B},\boldsymbol{\bar C}}  &\qquad \mathcal{J}(\bB, \bar\bC).  \label{opt_mse_2}
\end{align}

\begin{Proposition} \label{prop_mmse_2}
For fixed $(\bar\bG,\al)$, the optimal solution to \eqref{opt_mse_2} is given by
\begin{subequations} \label{bc}
\bal
\bB^{opt}=&\bD_1(\bD_2-\bD_1^H\bD_1)^{-1}[\bP\bQ\bH^H\bar\bG^H\bF^H]_{{\rm diag}}, \label{b}\\
\bar\bC^{opt}=&(\bD_2-\bD_1^H\bD_1)^{-1}[\bP\bQ\bH^H\bar\bG^H\bF^H]_{{\rm diag}}, \label{c}
\end{align}
\end{subequations}
where
\begin{subequations} \label{D1D2}
\begin{align}
\bD_1&= [\bF\bar\bG\bH]_{{\rm diag}}\in\mathbb{C}^{K\times K}, \label{D1}\\
\bD_2 &= [\bF\bar\bG\left(\bH\bQ\bH^H+\gm^2\bI\right)\bar\bG^H\bF^H+\sigma^2\al^{-2}\bI]_{{\rm diag}} \in\mathbb{C}^{K\times K}. \label{D2}
\end{align}
\end{subequations}
\end{Proposition}


\begin{proof}
The weighted sum MSE in \eqref{opt_mse_x1a} can be rewritten as
\be
&&\mathcal{J}(\bB,\bar\bC)   \notag \\
&\mathop =\limits^{}& \tr\left[\bW\bP\bQ\bP^T +\bW\bB \bQ\bB^H -2\Re\{\bW \bar\bC\bS+ \bW\bar\bC\bD_1\bQ\bB^H\}+ \bW\bar\bC \bD_2 \bar\bC^H\right] \label{48a}\\
 &\mathop =\limits^{}& \tr[\bW\bP\bQ\bP^T]+ \|\bW^{\frac{1}{2}}\bQ^{\frac{1}{2}}(\bB-\bD_1\bar\bC)\|_F^2 -2\Re\{\tr[ \bW \bar\bC\bS]\} \notag \\
&&+\tr[\bW\bar\bC^H (\bD_2-\bQ\bD_1^H\bD_1) \bar\bC] \label{48b}\\
 &\mathop =\limits^{}& \tr[\bW\bP\bQ\bP^T]+ \|\bW^{\frac{1}{2}}\bQ^{\frac{1}{2}}(\bB-\bD_1\bar\bC)\|_F^2 +\|\bW^{\frac{1}{2}}(\bD_3\bar\bC-\bD_3^{-1}\bS^H)\|_F^2   \notag \\
&& - \tr[\bW\bS^H\bD_3^{-2}\bS],  \label{71}
\ee
where \eqref{48a} follows by noting \eqref{D1D2} and letting
\be
\bS = [\bF\bar\bG\bH \bQ\bP^T]_{\rm diag},
\ee
and \eqref{48b} is obtained by factorizing w.r.t. $\bB$, and \eqref{71} follows by letting $\bD_3=(\bD_2-\bQ\bD_1^H\bD_1)^{\frac{1}{2}}$ and factorizing w.r.t. $\bar\bC$. In the above, $\bD_3=(\bD_2-\bQ\bD_1^H\bD_1)^{\frac{1}{2}}$ is well defined by noting that
\be
[\bD_2]_{ii}&\ge&[\bF\bar\bG\bH\bQ\bH^H\bar\bG^H\bF^H]_{ii} \label{22}\\
&\ge& q_i [\bF\bar\bG\bH]_{ii}^H [\bF\bar\bG\bH]_{\imath\imath} \label{23}\\
&=& q_i[\bD_1^H\bD_1]_{ii},\quad i=1,\cdots,K, \label{24}
\ee
where \eqref{22} follows by the definition in \eqref{D2}, and \eqref{24} follows \eqref{D1}. From \eqref{71}, we see that the weighted sum MSE is minimized when
\be
\bar\bC = \bD_3^{-2}\bS^H\quad {\rm and}\quad \bB=\bD_1\bar\bC,
\ee
which concludes the proof.
\end{proof}

From \eqref{bc}, we notice that
\be
\bB^{opt}=\bD_1\bar\bC^{opt}=[\bar\bC^{opt}\bF\bar\bG\bH]_{\text{diag}}. \label{16}
\ee
That is, $\bB^{opt}$ is exactly the self-interference weights in the filtered signal  $\bC\br$. This implies that,  with $\bB^{opt}$ given in \eqref{b},  the self-interference is perfectly canceled at each user end.

\subsection{Overall Iterative Algorithm}

The overall iterative algorithm is outlined by the following pseudo-code.

\medskip\noindent{\bf{Algorithm 1.}}

\begin{algorithmic}[1]
\STATE {\bf Initial}: $\bB=\bB_{0}$, $\bar\bC=\bar\bC_{0}$;
\WHILE{\text{the weighted sum MSE can be reduced by more than $\epsilon$}}
    \STATE Compute $\bar\bG$ and $\al$ using \eqref{10};
    \STATE Compute $\bB$ and $\bar\bC$ using \eqref{bc};
\ENDWHILE
\end{algorithmic}

The convergence of Algorithm $1$ is guaranteed because: 1) the weighted sum MSE decreases strictly monotonically in 
each step; and 2) the weighted sum MSE is non-negative. From Proposition $2.7.1$ in \cite{dimi99}, it can be further shown that this convergence point is a local optimum of the original problem \eqref{opt_mse}. The local optimum at which Algorithm $1$ stops in general depends on the initial condition $(\bB_0,\bar\bC_0)$. We shall discuss the choice of the initial condition and the speed of convergence in Section V.

\subsection{Further Discussion}

So far, our discussions have been constrained to the case that every user employs a same predetermined constant power in transmission. This constant power constraint can be relaxed to a maximum power constraint by reformulating Problem \eqref{opt_mse} as
\begin{subequations}  \label{opt_mse_new}
\begin{align}
\mathop\text{minimize}_{\boldsymbol{G},\boldsymbol{B},\boldsymbol{C},\bQ}  &\qquad \mathbb{E}\left\|\bW^{\frac{1}{2}}\left((\bP+\bB) \bx - \bC \br\right)\right\|^2  \\
\text{subject to} &\qquad  \text{Tr}\left[\bG\left(\bH\bQ\bH^H+\gm^2\bI\right)\bG^H\right]\le P_r \label{21b}\\
& \qquad q_i \leq Q_i,\quad  i=1,\cdots,K, \label{21c}\\
& \qquad \bB, \bC\text{ and } \bQ \text{ are diagonal}.
\end{align}
\end{subequations}
The above problem formulation allows power control among users, which leads to better performance. 

Following the spirit of our proposed iterative approach, we can approximately solve Problem \eqref{opt_mse_new} similarly to Algorithm $1$ by adding one more step in iteration. Specifically, besides the two steps of iteration in Algorithm $1$, we introduce one more step, namely, to optimize $\{q_i\}$ for fixed $(\bar\bG,\bB,\bar\bC)$. Note that, for any fixed $(\bar\bG,\bB,\bar\bC)$, both the cost function and the constraints \eqref{21b} and \eqref{21c} are linear functions of $\{q_i\}$, implying that Problem \eqref{opt_mse_new} reduces to a linear program of $\{q_i\}$, which is solvable by standard optimization tools. Furthermore, by noting that the feasible region of a linear program is a polytope and that the optimal solution of the linear program must occur at a vertex of this polytope, we can show that the optimal $\{q_i\}$ for Problem \eqref{opt_mse_new} is either $0$ or $Q_i$, i.e., each user either employs maximum power in transmission or keeps silent. We will return to this issue of user power control later in Sections IV and V.

\section{Weighted Sum Rate Maximization}

In this section, we study weighted sum rate maximization, which is another widely used in optimizing the performance of wireless networks. We show that the weighted sum rate maximization problem can be reformulated as an iterated weighted sum MSE minimization problem, and therefore, the techniques in Section III can be directly applied here.

\subsection{Problem Formulation}

For reasons to be clarified later, we now use a diagonal matrix $\bDt$, instead of $\bB$, to represent the weights of self-interference to be canceled. Note that letting $\bDt={\bf0}$ yields the conventional non-PNC problem. Then, after self-interference cancellation, the signal becomes $\br-\bDt\bx$.
To simplify the index mapping of the received signal vector, let
\be  \label{zr}
\bz=\bP^T (\br-\bDt\bx).
\ee
Using \eqref{r}, the $i$th element of $\bz$ is given by
\bal \label{z_i}
z_i=\bp_i^T\bF\bG\bh_i x_i+ \sum_{\ell\neq i}\bp_i^T\bF\bG\bh_\ell x_\ell  +\bp_i^T\bF\bG\bu+\bp_i^T\bv - \bp^T_i\bDt\bx,
\end{align}
where $\bh_i$ is the $i$th column of $\bH$. In this way, $z_i$ is the received signal of user $\pi(i)$ for the recovery of the message from user $i$.
Thus, the achievable rate of user $i$ is given by
\be \label{rate_x}
 R_i =\frac{1}{2}\log \left(1+\mathsf{SINR}_i\right),
\ee
where
\be  \label{sinr}
\mathsf{SINR}_i = \frac{q_i|\bp_i^T\bF\bG\bh_i|^2} {|\bp_i^T(\bF\bG\bH-\bDt)\bQ|^2-q_i|\bp_i^T\bF\bG\bh_i|^2+\gamma^2|\bp_i^T\bF\bG|^2 +\sigma^2}.
\ee
The factor $1/2$ is due to the two-phase transmission.
Our purpose is to maximize the weighted sum rate under the power constraint of the relay. This optimization problem is formulated as
\begin{subequations}  \label{max_rate}
\begin{align}
\mathop\text{maximize}_{\boldsymbol{G}, \bDt}  &\qquad \sum_{i=1}^K t_i R_i \label{sum_rate}\\
\text{subject to} &\qquad  \tr\left[\bG\left(\bH\bQ\bH^H+\gm^2\bI\right)\bG^H\right]\leq P_r, \label{opt_rate_b}
\end{align}
\end{subequations}
where $\bT=\diag\{t_1,\cdots,t_K\}$. The problem \eqref{max_rate} is non-convex w.r.t. $(\bG,\bDt)$ and thus is difficult to solve directly.

\subsection{Conversion to Weighted Sum-MSE Minimization}

Recall that we assume Gaussian inputs: $x_i\sim \mathcal{CN}(0,q_i)$, $i=1,\cdots,K$, i.e.,
\be \label{px}
p(x_i)=\frac{1}{\pi q_i}e^{-\frac{|x_i|^2}{q_i}},\quad i=1,\cdots,K.
\ee
From the equivalent transmission \eqref{z_i}, the conditional distribution can be readily obtained as
\be \label{prx}
p(z_i|x_i) = \frac{1}{\pi\Sm_i'}e^{-\frac{|z_i-\bp_i^T\bF\bG\bh_ix_i|^2}{\Sm_i'}},\quad i=1,\cdots,K,
\ee
where
\be
\Sm_i'&=& \bp_i^T\bF\bG(\bH\bQ\bH^H-q_i\bh_i\bh_i^H)\bG^H\bF^H\bp_i + \bp_i^T\bDt\bQ\bDt^H\bp_i \notag\\
 &&+ \gm^2\bp_i^T\bF\bG\bG^H\bF^H\bp_i + \sm^2\bp_i^T\bp_i.
\ee
With Bayes' rule, we can obtain the {\em a posteriori} distribution $p(x_i|z_i)$, which is also Gaussian and is given by
\be \label{pxr}
p(x_i|z_i)=\frac{1}{\pi\Sm_i}e^{-\frac{|x_i-\om_iz_i|^2}{\Sm_i}},\quad i=1,\cdots,K.
\ee
where $\om_i$ is a scaling coefficient to be determined,  $\om_iz_i$ represents the conditional mean, and $\Sm_i$ represents the conditional variance. From \cite{kay01}, the {\em a posteriori} mean and variance are respectively given by
\begin{subequations}
\bal
\mathbb{E}[x_i|z_i]&=\mathcal{C}_{x_iz_i}\mathcal{C}_{z_iz_i}^{-1}z_i,\label{mu}\\
\mathcal{C}_{x_ix_i|z_i}&=\mathcal{C}_{x_ix_i}-\mathcal{C}_{x_iz_i} \mathcal{C}_{z_iz_i}^{-1}\mathcal{C}_{z_ix_i}, \label{Sigma}
\end{align}
\end{subequations}
where the involved covariances are given by
\begin{subequations}
\be
\mathcal{C}_{x_ix_i}&=&q_i,\\
\mathcal{C}_{x_iz_i}&=& q_i\bh_i^H\bG^H\bF^H\bp_i, \label{crx}\\
\mathcal{C}_{z_iz_i}&=&\bp_i^T\bF\bG\left(\bH\bQ\bH^H+\gm^2\bI\right) \bG^H\bF^H\bp_i-\bp_i^T\bDt\bQ\bDt^H\bp_i + \sigma^2.  \label{crr}
\ee
\end{subequations}
Thus, we have
      \bal 
       \om_{\pi(i)} = \mathcal{C}_{x_iz_i} \mathcal{C}_{z_iz_i}^{-1}\quad \text{and}\quad      \Sigma_{\pi(i)}= \mathcal{C}_{x_ix_i|z_i},\ i=1,\cdots,K \label{Smi}
      \end{align}

From the Blahut-Arimoto algorithm in Lemma 13.8.1 of \cite{cover_bk}, the rate in \eqref{rate_x} can be written as
\be 
R_i&=&\frac{1}{2}\mathbb{E}_{x_i,z_i}\left[ \log\frac{p(x_i|z_i)}{p(x_i)}\right]\\
&=& \max_{\phi(\cdot|\cdot)} \frac12 \mathbb{E}_{x_i,z_i}\left[ \log\frac{\phi(x_i|z_i)}{p(x_i)}\right], \quad i=1,\cdots,K, \label{rate_cover}
\ee
where the expectation is taken over the joint distribution of $x_i$ and $z_i$, given by $p(x_i, z_i) = p(x_i)p(z_i|x_i)$ given in \eqref{px} and \eqref{prx}, and $\phi(\cdot|\cdot)$ is an arbitrary distribution of $x_i$ conditioned on $z_i$. It is known that the optimal choice of $\phi(x_i|z_i)$ follows the Gaussian distribution in the form of \eqref{pxr}. Then, with \eqref{pxr} and \eqref{px}, \eqref{rate_cover} is converted to be
\be  \label{xxx}
{R}_i &=& \max_{w_i, \Sm_i} \frac{1}{2} \mathbb{E}_{x_i,z_i}\bigg( \log\frac{q_i}{\Sm_i} -\frac{|x_i-\om_iz_i|^2}{\Sm_i} + \frac{|x_i|^2}{q_i}\bigg),
\ee
with the optimal $w_i$ and $\Sm_i$ given by \eqref{Smi}. Note that similar conversions have been previously used in \cite{cioffi08,xianda10} for optimizing beamforming vectors in broadcast channels. Based on this conversion, we next establish a relation between the weighted sum rate maximization problem \eqref{max_rate} and the weighted sum MSE minimization problem \eqref{opt_mse}.

Plugging \eqref{zr} into \eqref{xxx}, the sum rate is rewritten as
\bal  \label{rate_new}
&\sum_{i=1}^K t_i {R}_i  \notag \\
 =&  \max_{\{w_i\}, \{\Sm_i\}}  -\frac{1}{2}\left\{\eptn_{\bx,\bz}\left\|\bT^{\frac12}\hat\bSm^{-\frac{1}{2}}[(\bI+\hat\bOm\bP^T\bDt)\bx- \hat\bOm\bP^T\br]\right\|^2 + \sum_{i=1}^K t_i \log\frac{\Sm_i}{q_i}-\sum_{i=1}^K t_i \right\}\\
=& \max_{\{w_i\}, \{\Sm_i\}} -\frac{1}{2}\left\{\eptn_{\bx,\bz}\left\|\bT^{\frac{1}{2}}\bSm^{-\frac{1}{2}}[(\bP+\bOm\bDt)\bx-\bOm\br]\right\|^2 + \sum_{i=1}^K t_i \log\frac{\Sm_i}{q_i}-\sum_{i=1}^K t_i\right\}, \label{sum_rate_27}
\end{align}
where
\bseq
\be
\hat\bSm &=& \diag\{\Sm_1,\cdots,\Sm_K\},\\
\bSm &=& \diag\{\Sm_{\pi^{-1}(1)},\cdots,\Sm_{\pi^{-1}(K)}\}=\bP\hat\bSm\bP^T,\\
\hat\bOm &=& \diag\{\om_1,\cdots,\om_K\},\\
\bOm &=& \diag\{\om_{\pi^{-1}(1)},\cdots,\om_{\pi^{-1}(K)}\}=\bP\hat\bOm\bP^T. \label{28}
\ee
\eseq
The expectation taken over the joint distribution of $\bx$ and $\bz$ is equivalent to that taken over the joint distribution of  $\bx$, $\bu$, and $\bv$, where the noise vectors $\bu$, and $\bv$ were defined in \eqref{r}. In \eqref{sum_rate_27}, $\eptn_{\bx,\bz}\|\bSm^{-\frac{1}{2}}\left((\bP+\bOm\bDt)\bx-\bOm\br\right)\|^2$ is the same as the weighted sum MSE in \eqref{opt_mse_a} by letting\footnote{To avoid confusion in establishing this relation, we use $\bDt$, instead of $\bB$, to denote the weights of self-interference in the case of weighted sum rate maximization.}
\be  \label{wbc}
\bW^{}=\bT\bSm^{-1},\quad \bB=\bOm\bDt,\quad {\rm and}\quad
\bC=\bOm.
\ee
For fixed $\bOm$ and $\bSm$, the optimization of $\bG$ is exactly the same as that for the weighted sum MSE minimization problem in Section III.B. For fixed $\bG$, the optimal $\bOm$ and $\bSm$ can be determined by the MMSE estimator of the transmission in \eqref{z_i}, which will be presented explicitly in the following subsection.

\subsection{Iterative Algorithm}

Similarly to \eqref{al_mmse}, we define
\be \label{beta}
\bar\bG\triangleq \al^{-1}\bG,\quad \bar\bDt = \al^{-1}\bDt, \quad {\rm and}\quad \bar \bOm\triangleq \al \bOm,
\ee
where $\al$ still represents the scaling factor to meet the relay power constraint. With \eqref{beta} and the received signal vector at the users $\br$ in \eqref{r}, the weighted sum rate in \eqref{sum_rate_27} is expanded as
\bal 
\mathcal{R}&(\bar\bG, \al, \bar\bDt, \bar\bOm, \bSm)  \notag \\
= &-\frac{1}{2} \tr\left[\bT\bSm^{-1}\left(\bP\bQ\bP^T+\bar\bOm\bar\bDt\bQ\bar\bDt^H \bar\bOm^H -2\Re\{\bar\bOm\bF\bar\bG\bH\bQ(\bP+\bar\bOm\bar\bDt)^H\}\right.\right. \notag \\ &\left.\left.+\bar\bOm\bF\bar\bG\left(\bH\bQ\bH^H+\gm^2\bI\right)\bar\bG^H\bF^H\bar\bOm^H +\sm^2\al^{-2}\bar\bOm \bar\bOm^H\right)\right] - \frac{1}{2}\sum_{i=1}^K t_i\log\frac{\Sm_i}{q_i} +\frac12 \sum_{i=1}^K t_i. \label{rate_expand}
\end{align}
Then, based on \eqref{sum_rate_27}, we reformulate the optimization problem \eqref{max_rate} as
\begin{subequations}  \label{max_rate_x}
\begin{align}
\mathop \text{maximize}_{\mathclap{\boldsymbol{\bar G}, \al, \boldsymbol{{\bar\Delta}}, \boldsymbol{{\bar\Omega}}, \boldsymbol{\Sigma}}}  &\qquad \mathcal{R}(\bar\bG, \al, \bar\bDt, \bar\bOm, \bSm)\\
\text{subject to} &\qquad  \al^2\tr\left[\bar\bG\left(\bH\bQ\bH^H+\gm^2\bI\right)\bar\bG^H\right]\leq P_r \label{relay_power}.
\end{align}
\end{subequations}
Noting the similarity to problem \eqref{opt_mse}, we develop an iterative algorithm to solve problem \eqref{max_rate_x} as follows.

\subsubsection{Optimal $(\bar\bG,\al)$ for fixed $(\bar\bDt, \bar\bOm, \bSm)$} \label{D}

For fixed $(\bar\bDt, \bar\bOm, \bSm)$, problem \eqref{max_rate_x} is the same as the one in \eqref{opt_mse_1} by letting $\bW=\bT\bSm^{-1}$, $\bB=\bar\bOm\bar\bDt$ and $\bar\bC=\bar\bOm$ except for some additive constants. Thus, from Proposition \ref{prop_mmse_1}, the optimal precoder can be immediately written as $\bG^{opt} = \al\bar\bG^{opt}$ with
\bseq  \label{Gnopt_rate}
\be
\bar\bG^{opt}& =& \left(\frac{\sm^2}{P_r}\tr[\bT\bSm^{-1}\bar\bOm \bar\bOm^H]\bI + \bF^H\bar\bOm^H \bT\bSm^{-1}\bar\bOm\bF \right)^{-1} \bF^H\bar\bOm^H\bT \bSm^{-1} \notag\\
&&\times(\bP+\bar\bOm\bar\bDt)\bQ\bH^H \left(\bH\bQ\bH^H+\gm^2\bI\right)^{-1} \label{g2}\\
\al^{opt} &=& P_r^{\frac{1}{2}}\left(\text{Tr}\left[\bar\bG^{opt} \left(\bH\bQ\bH^H+\gm^2\bI\right)(\bar\bG^{opt})^H\right]\right)^{-\frac{1}{2}}. \label{al2}
\ee
\eseq

\subsubsection{Optimal $(\bar\bDt, \bar\bOm, \bSm)$ for fixed $(\bar\bG,\al)$} \label{D}

For fixed $(\bar\bG,\al)$, we aim to find the optimal $3$-tuple $(\bar\bDt, \bar\bOm, \bSm)$ that maximizes $\mathcal{R}(\bar\bDt, \bar\bOm, \bSm)$.
We first determine the optimal $\bar\bDt$. From \eqref{wbc} and \eqref{beta}, we see that
$\eptn_{x_i,z_i}\left\|\bT^{\frac{1}{2}}\bSm^{-\frac{1}{2}}[(\bP+\bOm\bDt)\bx-\bOm\br]\right\|^2$ in \eqref{sum_rate_27} is equivalent to \eqref{opt_mse_a} by replacing $\bW$ with $\bT\bSm^{-1}$, $\bB$ with $\bOm\bDt$ and $\bC$ with $\bOm$. Together with the fact that the optimal $\bB$ for \eqref{opt_mse_2} is given in \eqref{16}, i.e., $\bB^{opt}=[\bar\bC^{opt}\bF\bar\bG\bH]_{\text{diag}}$, the optimal $\bar\bDt$ is given by
\be \label{bopt}
\bar\bDt^{opt} = [\bF\bar\bG\bH]_{\rm diag}.
\ee
With \eqref{bopt}, we obtain $\bDt^{opt}=\al\bar\bDt^{opt}=[\bF\bG\bH]_{\rm diag}$ which consists of the self-interference weights in the received signal $\br$. This means that the self-interference is perfectly canceled at the receiver ends. (Since self-interference is known precisely, it is rather obvious it should be completely canceled before detection.)
For fixed $(\bar\bG,\al,\bar\bDt)$, the optimal $(\bar\bOm, \bSm)$ is determined by \eqref{Smi}, i.e., the MMSE estimator of the transmission in \eqref{z_i}. 


\subsubsection{Overall iterative algorithm}

The weighted sum rate optimization problem \eqref{max_rate} can be solved by iteratively solving the above two subproblems. The procedure is outlined in the following algorithm.

\medskip\noindent{\bf{Algorithm 2.}}

\begin{algorithmic}[1]
\STATE {\bf Init}: $\bar\bDt=\bDt_0$, $\bar\bOm=\bOm_0$, $\bSm=\bSm_0$;
\WHILE{\text{the weighted sum rate can be improved by more than $\delta$}}
    \STATE Compute $\bar\bG$ and $\al$ using \eqref{Gnopt_rate};
    \STATE Compute $\bar \bDt$, $\bar\bOm$ and $\bSm$ using \eqref{bopt} and \eqref{Smi};
\ENDWHILE
\end{algorithmic}

The convergence of Algorithm $2$ is guaranteed, as the weighted sum rate is bounded and strictly monotonically increases in the iterative process. The convergence point depends on the initial point $(\bar\bDt_0,\bar\bOm_0,\bSm_0)$. We will discuss the choice of the initial point of Algorithm $2$ in Section V.

\subsection{Further Discussion}

It is interesting to compare Algorithms $1$ and $2$. We see that Algorithm $2$ can be treated as a weighted sum MSE minimization algorithm with varying weights, since the weight matrix $\bW$ remains constant in Algorithm $1$, but the corresponding ``weight matrix'' $\bT\bSm^{-1}$ varies in Algorithm $2$. The similarities between the two algorithms eventually lead to similar asymptotic behaviors in the extreme SNR regimes, as will be shown in the next section.

Moreover, similarly to the case of weighted sum MSE minimization, we may also consider the user power control for the weighted sum rate maximization problem. Specifically, Problem \eqref{max_rate_x} can be reformulated as 
\begin{subequations}  \label{max_rate_x_new}
\begin{align}
\mathop \text{maximize}_{\mathclap{\boldsymbol{\bar G}, \al, \boldsymbol{{\bar\Delta}}, \boldsymbol{{\bar\Omega}}, \boldsymbol{\Sigma},\bQ}}  &\qquad \mathcal{R}(\bar\bG, \al, \bar\bDt, \bar\bOm, \bSm)\\
\text{subject to} &\qquad  \al^2\tr\left[\bar\bG\left(\bH\bQ\bH^H+\gm^2\bI\right)\bar\bG^H\right]\leq P_r, \label{relay_power}\\
& \qquad q_i \leq Q_i,\quad  i=1,\cdots,K,\\
&\qquad \bQ\text{ is diagonal}. 
\end{align}
\end{subequations}
Again, the new problem can be approximately solved in an iterative fashion by adding one more iteration step to Algorithm $2$. This extra step is to optimize $\{q_i\}$ by fixing the other variables, which is solvable by standard linear programming. We omit the details here.

\section{Asymptotic Analysis}

Algorithms $1$ and $2$ only guarantee local optima of the weighted sum MSE minimization and weighted sum rate maximization problems. In this section, we carry out asymptotic analysis and show that, with proper initialization, the proposed iterative algorithms are asymptotically optimal in the low and high SNR regimes.
For ease of discussion, we assign equal weights to the weighted sum MSE minimization problem, i.e., $\bW=\bT$. We will see that the asymptotic solutions to the weighted sum MSE minimization and weighted sum rate maximization allows unified expressions. 

\subsection{Low-SNR Analysis}

We start with the low-SNR case. We focus on the limit where the noise levels $\sm^2$ and $\gm^2$ tend to infinity, i.e., $\sm^2,\gm^2\rightarrow +\infty$.
The main result is summarized as follows; the proof is given in Appendix \ref{theorem_lowp}.

\begin{framed}
\begin{Theorem}  \label{theorem_low}
In the limit of $\sm^2, \gm^2\rightarrow +\infty$, the asymptotically optimal precoders for the weighted sum MSE minimization in \eqref{opt_mse} and the weighted sum rate maximization in \eqref{max_rate}, with and without PNC, are identical and can be expressed as
\be \label{44}
\bG^0 = \al\bar\bG^0,
\ee
where $\bar\bG^0$ is such that $vec(\bar\bG^0)$ is an eigenvector  corresponding to the maximum eigenvalue of
\be  \label{psi}
\bPsi &=& \sum_{\ell=1}^K q_\ell^2 w_\ell  (\bh_\ell\bh_\ell^H)^T\otimes(\bF^H\bp_\ell\bp_\ell^T\bF),
\ee
and scalar $\al$ is such that the precoder $\bG^0$ satisfies the power constraint with equality at the relay.
\end{Theorem}
\end{framed}

At low SNR, the optimal precoder is identical with and without PNC in the limit. This is not surprising because when the noise dominates the received signal, the benefit of  self interference cancellation is marginal.


\begin{Proposition}  \label{prop_low_converge}
As $\sm^2,\gm^2\rightarrow +\infty$, Algorithms $1$ and $2$ converge to the same global optimum $\bG^0$ given in \eqref{44}.
\end{Proposition}

The proof of Proposition \ref{prop_low_converge} is given in Appendix \ref{prop_low_converge_p}. According to Proposition \ref{prop_low_converge}, the point of convergence of Algorithms $1$ and $2$ is not sensitive to the initial condition in the low SNR regime.

\subsection{High-SNR Analysis}

In the high SNR regime, we are interested in the limit of $\sm^2,\gm^2\rightarrow 0$.
Then, we discuss the degree of freedom (DoF) that the system could achieve with the setup of different numbers of relay antennas and users. 

First, for $N\ge K$, the asymptotically optimal precoders are described as follows, where the proof is given in Appendix \ref{theorem_highp}.

\begin{framed}
\begin{Theorem}  \label{theorem_high}
Suppose $N\ge K$. In the limit of $\sm^2,\gm^2\rightarrow 0$, the asymptotically optimal precoders for weighted sum MSE minimization and weighted sum rate maximization, with and without PNC, are identical and can be expressed as
\be  \label{asymp_high}
\bG^\infty = \bF^{\dag}\bC^{-1}(\bP+\bB)\bH^{\dag},
\ee
where $\bC\in \mathbb{C}^{K\times K}$ is diagonal, and $\bB\in \mathbb{C}^{K\times K}$ is an all-zero matrix in the non-PNC case and is a diagonal matrix in the PNC case.
\end{Theorem}
\end{framed}

It has been shown in \cite{wms12jsac} that the precoder in \eqref{asymp_high} forces all the interference to zero, and hence is referred to as the zero-forcing precoder. Theorem \ref{theorem_high} reveals that zero-forcing precoding is asymptotically optimal when the relay has no fewer antennas than the number of users (i.e., $N\ge K$). In fact, for the case with PNC, full degrees of freedom can be achieved when the relay has $N=K-1$ antennas only. But the optimal solution in this case is not in the form of \eqref{asymp_high}. 

Next, we discuss the asymptotically optimal precoder for $N=K-1$. 
\begin{Proposition}  \label{prop_N=K-1}
Suppose $N=K-1$. In the limit of $\sm^2,\gm^2\rightarrow 0$, the asymptotically optimal precoders for the weighted sum MSE minimization and the weighted sum rate maximization, are identical and can be expressed as
\be \label{g_new}
\bG^{\infty}= vec (\bn_{\bS}) \al,
\ee
where $\al$ is determined by the power constraint at the relay, and $\bn_{\bS}$ is the null vector of $\bS$. The columns of $\bS\in \mathcal C^{(K-1)\times K(K-2)}$ are given by $\bff_{j} \otimes \bh_{i}$ for $i,j=1,\cdots,K,\ \text{and } j\ne i,\ j\ne \pi(i)$.
\end{Proposition}

\begin{proof}
As concluded from Theorem $2$, zero forcing relaying is required to achieve full DoF. Each user can see the signals from the desired user and itself only. This can be achieved by
\be
\bff_{j}^T \bG \bh_{i} =  (\bh_{i} \otimes \bff_{j})^T vec(\bG)=0,\quad \text{for } i,j=1,\cdots,K,\ \text{and } j\ne i,\ j\ne \pi(i),
\ee
where $\bff_{j}\in
\mathbb{C}^{N\times 1}$ be the downlink channel vector from the relays to user $j$, i.e., the $j$th row of $\bF$,
or equivalently in a matrix form as
\be \label{gS=0}
\bS^T vec(\bG) = {\bf0}.
\ee

To ensure full DoF, there must exist non-zero $\bG$ that satisfies \eqref{gS=0}, or equivalently, the null space of $\bS$ does not only consist of the zero vector. By definition, $\bS$ is determined by the channel matrices $\bF$ and $\bH$. For randomly generated $\bF$ and $\bH$, $\bS$ is of full rank with probability one. Therefore, a non-trivial null space of $\bS$ requires that the number of relays exceeds the column rank of $\bS$, i.e., $(K-1)^2 > K(K-2)$. Thus, we conclude \eqref{g_new}.
\end{proof}

So far, we have assumed $N\ge K-1$ to provide enough DoF for zero-forcing relaying. Otherwise, the relay is incapable of forcing every inter-user interference to zero simultaneously. For $N < K-1$, we can circumvent this issue by user scheduling. Specifically, we deactivate some users to ensure that the number of the active users does not exceed $N$ for the non-PNC case and $N+1$ for the PNC case. Then, once the scheduling strategy is given, the precoder proposed in this section can be directly applied. The design of the scheduling strategy is not the focus, thus is omitted in this paper.


\begin{Proposition}  \label{prop_high_converge}
Suppose $N\ge K$. As $\sm^2,\gm^2\rightarrow 0$, Algorithms $1$ and $2$ converge to $\bG=\bF^{\dag}\bC_0^{-1}(\bP+\bB_0)\bH^{\dag}$ for any initial values $\bB=\bB_0$ and $\bC=\bC_0$.
\end{Proposition}

Proposition \ref{prop_high_converge} suggests that, in the high SNR regime, the convergence points of Algorithms $1$ and $2$ highly depend on the initial conditions.  Therefore, it is necessary to carefully choose $\bB_0$ and $\bC_0$ in the high SNR regime, as detailed below. Together with Proposition \ref{prop_low_converge}, it is suggested that Algorithm $1$ and $2$ should be initialized by high-SNR aysmptotically optimal/suboptimal solutions over all SNR regimes. 

Note that the high-SNR asymptotically optimal precoder for $N=K-1$ is determined by \eqref{g_new}, where $\al$ is determined by the relay power constraint. Thus, we could use the high-SNR asymptotically optimal precoder as the initial values for the iterative algorithms. 

Next, we consider the results for $N\ge K$.

\subsubsection{Weighted sum MSE minimization}

Plugging in the optimal precoder \eqref{asymp_high} and ignoring the high-order infinitesimals, we can rewrite the weighted sum MSE minimization problem in \eqref{opt_mse} as
\begin{subequations} \label{pmse}
\begin{align}
\mathop\text{minimize}_{\boldsymbol{B},\boldsymbol{C}}  &\qquad \tr\left[\bW(\gm^2(\bP+\bB)(\bH^{H}\bH)^{-1}(\bP+\bB)^H +\sm^2 \bC\bC^H)\right] \label{14a}\\
\text{subject to} &\qquad  \text{Tr}[(\bF\bF^{H})^{-1}\bC^{-1}(\bP+\bB)\bQ(\bP+\bB)^H \bC^{-H}]\le P_r \label{14b}\\
& \qquad \bB\text{ and } \bC \text{ are diagonal}. \label{diag}
\end{align}
\end{subequations}





Problem \eqref{pmse} is in general non-convex and is difficult to solve. In the following, we propose a suboptimal solution to iteratively optimize $\bB$ and $\bC$. Define the index mapping $\jmath=\pi(\imath)$, $\imath=1,\cdots,K$.

\begin{Proposition}  \label{prop_high_mmse_2}
(i) Given $\{c_\jmath\}$, the problem \eqref{pmse} is a convex problem of $\{b_\jmath\}$ with the optimal solution
\be
b_\jmath = -\frac{\gm^2 w_\jmath h_{\imath\jmath} + \lb q_{\jmath} f_{\jmath,\pi(\jmath)}(c_\jmath^*)^{-1}c^{-1}_{\pi(\jmath)}}{\gm^2w_{\jmath} h_{\jmath\jmath} + \lb q_{\jmath} f_{\jmath\jmath}|c_\jmath|^{-2} },\quad \jmath = 1,\cdots,K,
\ee
where $\lb$ is a scalar to meet the relay power constraint,  $f_{ij}$ is element $(i,j)$ of $(\bF\bF^H)^{-1}$, and $h_{ij}$ is element $(i,j)$ of $(\bH^H\bH)^{-1}$.
(ii) Given $\{b_\jmath\}$, the optimal phases of $\{c_\jmath\}$ to the problem \eqref{pmse} are given by
\be \label{47}
\angle c_\jmath&=&\angle c_\imath-\angle f_{\jmath\imath}-\angle b_{\imath}-\pi,\quad \jmath = 1,\cdots,K.
\ee
(iii) Given $\{b_\jmath\}$ and the phases of $\{c_\jmath\}$ in \eqref{47}, problem \eqref{pmse} is convex in $\{|c_\jmath|^{-2}\}$.
\end{Proposition}

Proposition \ref{prop_high_mmse_2} is proved in Appendix \ref{prop_high_mmse_2p}. With Proposition \ref{prop_high_mmse_2}, we can readily develop an algorithm to iteratively optimize $\{b_\jmath\}$ and $\{c_\jmath\}$. We omit the details for simplicity.

The phase-aligned algorithm in \cite{wms12jsac} is the special case of pairwise traffic pattern assuming the uplink and downlink channel are reciprocal. The phase setting \eqref{47} generalizes the previous result by relaxing the assumptions of pairwise pattern and channel reciprocity.


\subsubsection{Weighted sum rate maximization}

Next, we investigate the high-SNR asymptotically optimal matrices $\bB$ and $\bC$ for the weighted sum rate maximization.
Plugging into the optimal form \eqref{asymp_high} and ignoring the high-order infinitesimals, we rewrite \eqref{max_rate} as
\begin{subequations} \label{pmse2}
\begin{align}
\mathop\text{maximize}_{\{b_\jmath\},\{c_\jmath\}}  &\qquad \frac{1}{2}\sum_{\jmath=1}^{K}  \log \frac{t_j q_j|c_\jmath|^{-2}}{\gamma^2 |c_\jmath|^{-2}\left[ h_{\jmath\jmath} |b_\jmath|^2+2\Re\{h_{\imath\jmath}^*b_\jmath\}+ h_{\imath\imath}\right]
+\sm^2} \label{60a}\\
\text{subject to} &\qquad  \text{Tr}[(\bF\bF^{H})^{-1}\bC^{-1}(\bP+\bB)\bQ(\bP+\bB)^H \bC^{-H}]\le P_r.  \label{60b}
\end{align}
\end{subequations}





Similarly to the case of the weighted sum MSE minimization, the joint optimization of $\bB$ and $\bC$ is difficult to solve. We next find a suboptimal solution by iteratively optimizing $\{b_\jmath\}$ and $\{c_\jmath\}$. Similarly to Proposition \ref{prop_high_mmse_2}, we have the following results, with the proof of which is given in Appendix \ref{prop_high_rate_2p}.


\begin{Proposition}  \label{prop_high_rate_2}
(i) Given $\{c_\jmath\}$, problem \eqref{pmse2} has a closed-form optimal solution $\{b_\jmath\}$. (ii) Given $\{b_\jmath\}$, the optimal phases of $\{c_\jmath\}$ to problem \eqref{pmse2} are given by \eqref{47}.
(iii) Given $\{b_\jmath\}$ and the phases of $\{c_\jmath\}$ in \eqref{47}, the problem \eqref{pmse} is convex in $\{|c_\jmath|^2\}$.
\end{Proposition}

\subsection{Further Discussion}

To summarize, our asymptotic analysis  reveals that the proposed iterative algorithms in Section III converge to the asymptotically optimal solution in \eqref{44} at low SNR, and this convergence is insensitive to the initial conditions. At high SNR, Algorithms $1$ and $2$  converge to the asymptotically optimal zero-forcing form in \eqref{asymp_high}, but could perform poorly depending on the initial conditions of $\bB$ and $\bC$. Therefore, in implementation, we set the initial values of Algorithms $1$ and $2$ to the high-SNR optimal/suboptimal solutions of $\bB$ and $\bC$ (cf. Propositions $4-7$).

It is also worth mentioning that all the theorems and propositions presented in this section literally hold for the case in which user power control is allowed. That is, all the results obtained in this section are directly applicable to the asymptotic solutions to problems \eqref{opt_mse_new} and \eqref{max_rate_x_new}. To see this, we only need to show that, in the asymptotic SNR regimes, maximum power transmission is desirable for every user. We first consider the high SNR regime. It is known that zero forcing is asymptotically optimal at high SNR, which implies that the higher the transmission power, the higher the weighted sum rate and the lower the weighted sum MSE (since the users do not interfere with each other). Therefore, provided $N \ge K-1$ (which allows zero-forcing), full power transmission at every user is asymptotically optimal in the high SNR regime.

Now we consider the low SNR regime. In this case, channel noise dominates the inter-user interference, implying that the higher transmission power, the higher achievable rate for each user (and also the lower MSE for each user). Therefore, maximum power transmission is also asymptotically optimal in the low SNR regime.

The asymptotic analysis provided in this section also sheds light on the convergence speed of the proposed iterative algorithms. Specifically, these algorithms converge very fast at high SNR, and eventually are stuck at the initial value when SNR goes to infinity. In the low SNR regime, the proposed algorithms reduce to the power iteration method (used for finding the maximum eigenvalue of a matrix). The convergence speed of power iteration depends on the ratio of the second largest eigenvalue against the largest eigenvalue, which depends on the specific channel realization. Roughly speaking, the proposed algorithms converge relatively slow at low SNR, and the convergence speed increases as SNR increases, as will be demonstrated in the next section.

\section{Numerical Results}

In this section, we evaluate the weighted sum MSE and the weighted sum rate of the proposed MIMO switching schemes. We assume that the maximum transmit
power levels of the relay and the users are the same (thus $Q_i=P_r = 1$, $i=1,\cdots K$), and the users transmit with maximum power. The noise levels at the relay and
at the users are the same (i.e., $\sm^2 = \gm^2$). Then, the transmit SNR is defined as $\mathsf{SNR}= 1/\sm^2 = 1/\gm^2$. The antenna and user settings are $N = K = 4$. We present the numerical results averaged over all permutations (there are $9$ different derangements for $N=4$).
We assume Rayleigh fading, i.e., the elements of $\bH$ and $\bF$ are independently drawn from $\mathcal{CN}(0,1)$. Each simulation point in the presented figures is obtained by averaging over $10^5$ random channel realizations. Note that the results of non-PNC schemes can be found in Appendix G, which will be used for comparison in simulation. 

The key findings are summarized as four observations.

\begin{figure}
\centering
\includegraphics[width=6in]{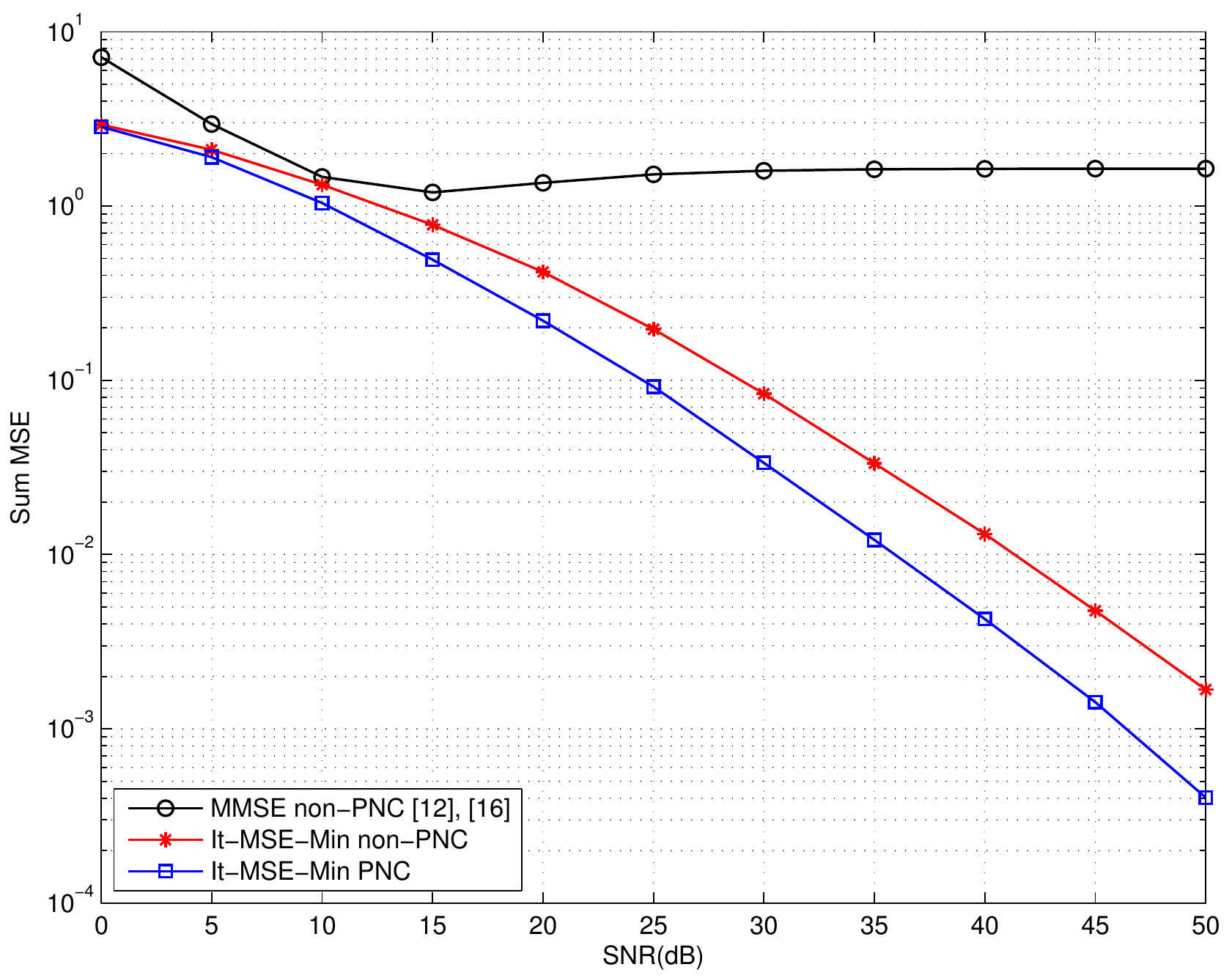}
\caption{The sum MSE performance is evaluated for different schemes when $N=K=4$. The sum MSE achieved by \cite{aut10,Joung10} saturates at high SNR. In contrast, the sum MSE of the proposed iterative schemes, which are labelled by ``It-MSE-Min non-PNC'' and ``It-MSE-Min PNC'', decreses as SNR increases. The one with PNC outperforms the one without PNC.}
\label{mse}
\end{figure}

\medskip\noindent{\bf\emph{Observation 1}}: {\em The proposed iterative sum MSE minimization (It-MSE-Min) algorithm  performs significantly better than the MMSE scheme in \cite{aut10,Joung10}.}

The MSE performance of the related schemes is illustrated in Fig.~\ref{mse}. The sum MSE of the MMSE scheme \cite{aut10,Joung10} saturates at high SNR (say, $\mathsf{SNR}>10$ dB), since this scheme optimizes the precoder at the relay only. In contrast, the proposed It-MSE-Min algorithm jointly optimizes both the precoder at the relay and the receive filters at the users, resulting in vanishing MSE at high SNR. Moreover, the sum MSE can be further reduced by exploiting the PNC technique. From Fig.~\ref{mse}, the PNC gain
is as significant as $6$ dB at the sum MSE of $10^{-2}$.

\begin{figure}
\centering
\includegraphics[width=6in]{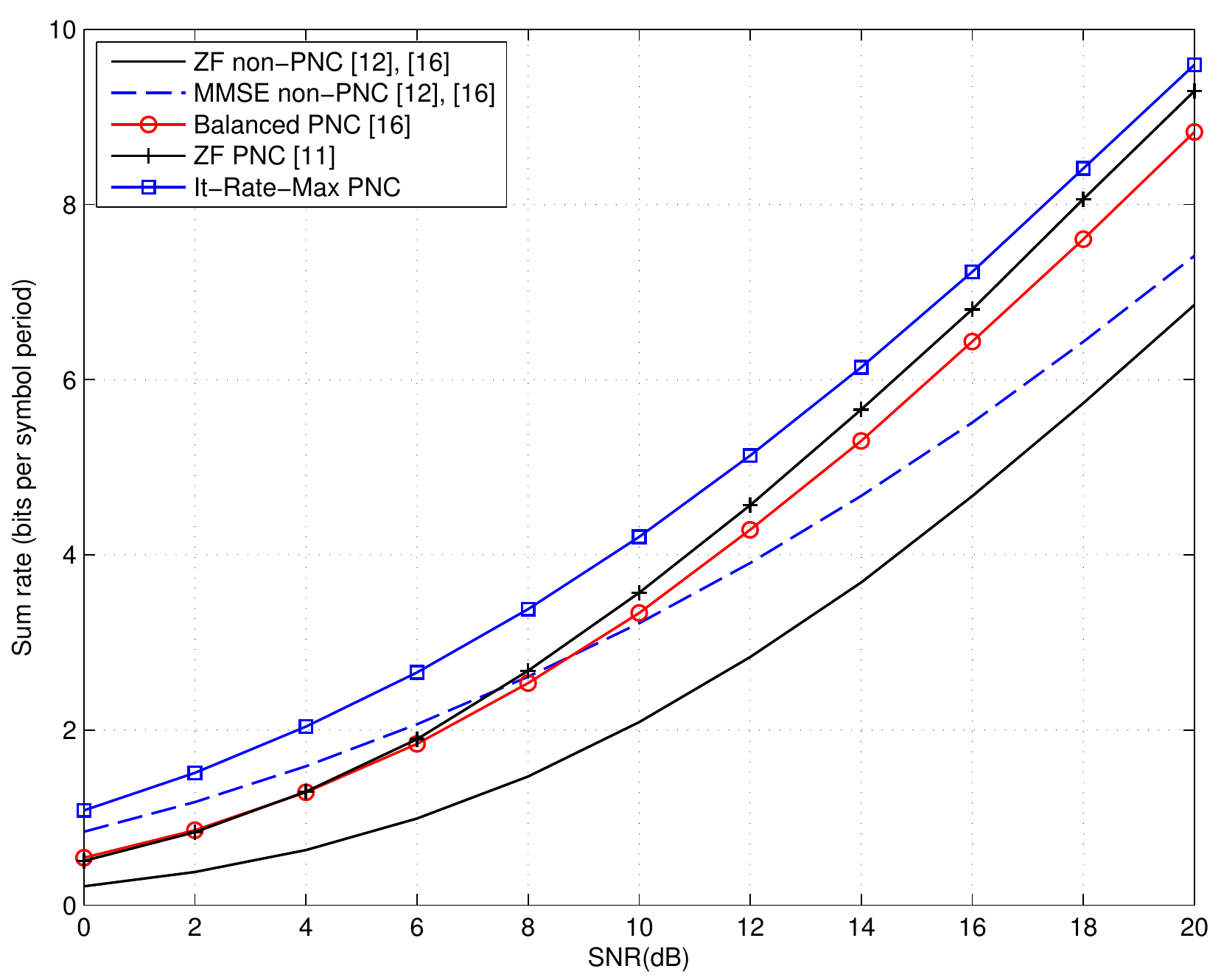}
\caption{The sum rate performance is evaluated for different schemes when $N=K=4$. The proposed iterative algorithm for sum rate maximization with PNC, i.e., ``It-Rate-Max PNC'', outperforms other schemes in the literature. ``ZF non-PNC'' and ``MMSE non-PNC'' were proposed in \cite{aut10,Joung10}; ``Balanced PNC'' was proposed in \cite{aut10} only for pairwise data exchange; ``ZF PNC'' was proposed in \cite{wms12jsac}.}
\label{rate}
\end{figure}

\medskip\noindent{\bf\emph{Observation 2}}: {\em The iterative sum rate maximization (It-Rate-Max) algorithm achieves significant throughput gains over the existing relaying schemes, such as ZF/MMSE relaying \cite{aut10,Joung10} and the network-coded relaying \cite{aut10,ey10}.}

\begin{figure}
\centering
\includegraphics[width=6in]{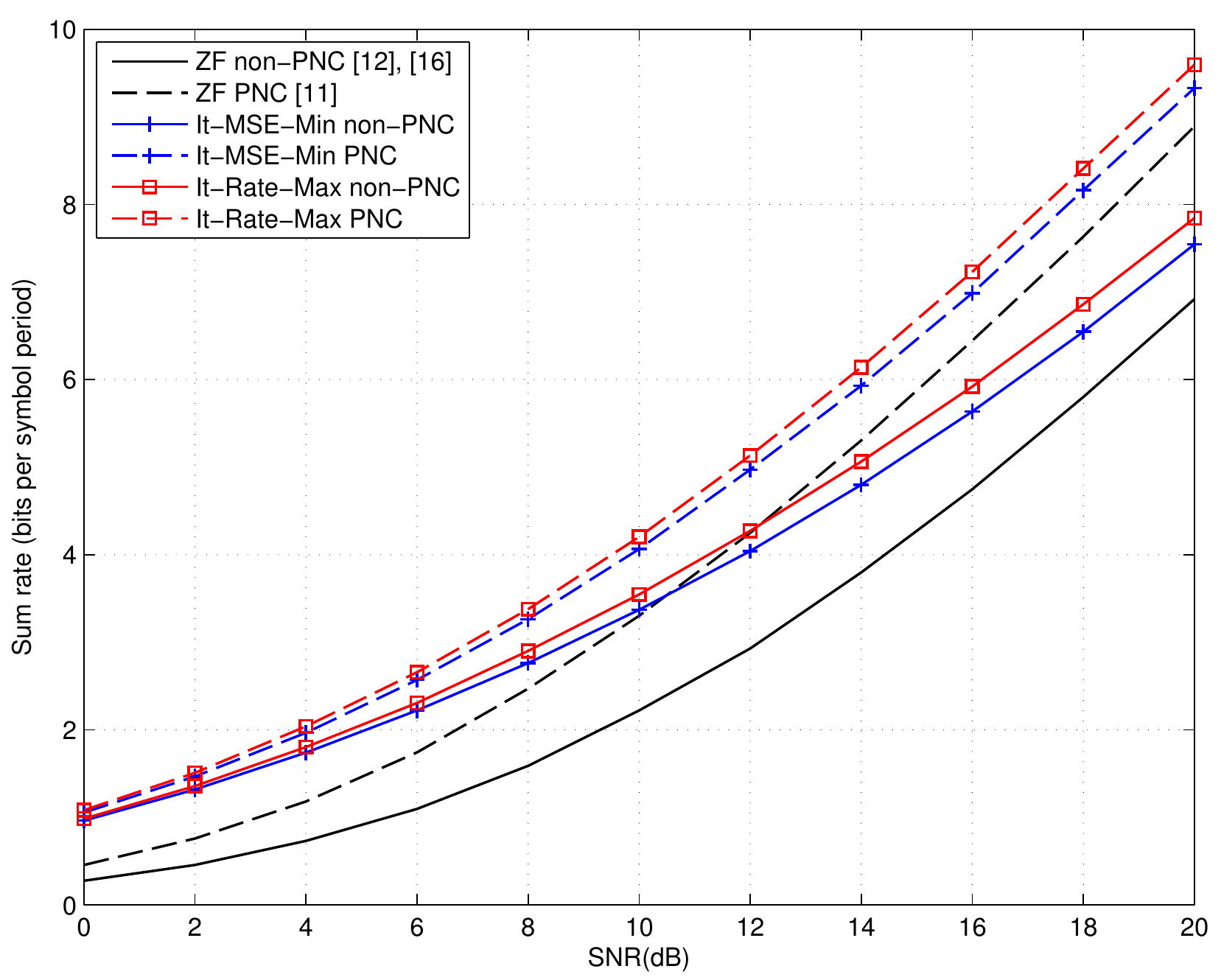}
\caption{The sum rate performance is evaluated for corresponding non-PNC and PNC schemes when $N=K=4$. In general, the schemes with PNC outperform the ones without PNC.
``ZF non-PNC'' was proposed in \cite{aut10,Joung10}; ``ZF PNC'' was proposed in \cite{wms12jsac}; ``It-MSE-Min'' and ``It-Rate-Max'' are proposed in Section III and IV, respectively, where their non-PNC schemes are derived in Appendix \ref{prop_non_pnc}.}
\label{nPNC_PNC}
\end{figure}

Fig.~\ref{rate} illustrates the throughput performance of various approaches including the proposed iterative sum rate maximization scheme with PNC (It-Rate-Max PNC), the zero-forcing scheme without PNC (ZF non-PNC) in \cite{aut10,Joung10}, the MMSE scheme without PNC (MMSE non-PNC) in \cite{aut10,Joung10}, the balanced PNC scheme  proposed in \cite{aut10},\footnote{A PNC scheme was proposed in \cite{ey10} as well, which used the same block-diagonalization technique as that in \cite{aut10}. However, the scheme in \cite{aut10} induced an extra step to balance the channel gain of each user, which outperformed the scheme in \cite{ey10}. Thus, we show the result of \cite{aut10} in Fig.~\ref{rate} only. Note that the proposed schemes in \cite{ey10} and \cite{aut10} are both for pairwise transmission only. Thus, the red curve in Fig.~\ref{rate} is calculated by averaging the results for symmetric derangements only.} and the zero-forcing scheme with PNC (ZF PNC) in \cite{wms12jsac}. From Fig.~\ref{rate}, the proposed It-Rate-Max PNC algorithm significantly outperforms the other schemes throughout the SNR range of interest. Specifically, the proposed algorithm outperforms the MMSE non-PNC scheme, especially in the high SNR regime, since the former utilizes the PNC technique and jointly optimizes the precoder and the receive filter. The proposed algorithm also outperforms the zero-forcing schemes in \cite{aut10,Joung10,wms12jsac}, since the latters suffer from noise enhancement. Furthermore, we also see that the proposed iterative rate-max scheme has roughly $1.5$ dB gain over the balanced PNC scheme in \cite{aut10} throughout the whole SNR range of interest.


\medskip\noindent{\bf\emph{Observation 3}}: {\em The PNC schemes achieve considerably higher throughputs than their corresponding non-PNC schemes, especially at medium and high SNR.}

Fig.~\ref{nPNC_PNC} illustrates the PNC gain for the proposed It-MSE-Min/It-Rate-Max approaches, as well as for the zero-forcing relaying schemes in \cite{aut10,Joung10,wms12jsac}. At low SNR, the proposed It-MSE-Min/It-Rate-Max algorithms with and without PNC, exhibit roughly the same throughput performance, which numerically verifies Theorem \ref{theorem_low}. At high SNR, the proposed schemes with PNC achieve about $6$ dB gain over the best non-PNC schemes (i.e., the It-Rate-Max scheme without PNC) at the sum rate of $8$ bits per symbol period.
The proposed It-MSE-Min/It-Rate-Max algorithms exhibit similar throughput performance at high SNR, either in the PNC case or the non-PNC case. This agrees well with the fact in Theorems $1$ and $2$ that the optimal precoders for It-MSE-Min/It-Rate-Max coincide, except that the asymptotically optimal diagonal matrices $\bB$ and $\bC$ are slightly different (cf. Proposition \ref{prop_high_converge}$-$\ref{prop_high_rate_2}).


\begin{figure}
\centering
\includegraphics[width=6in]{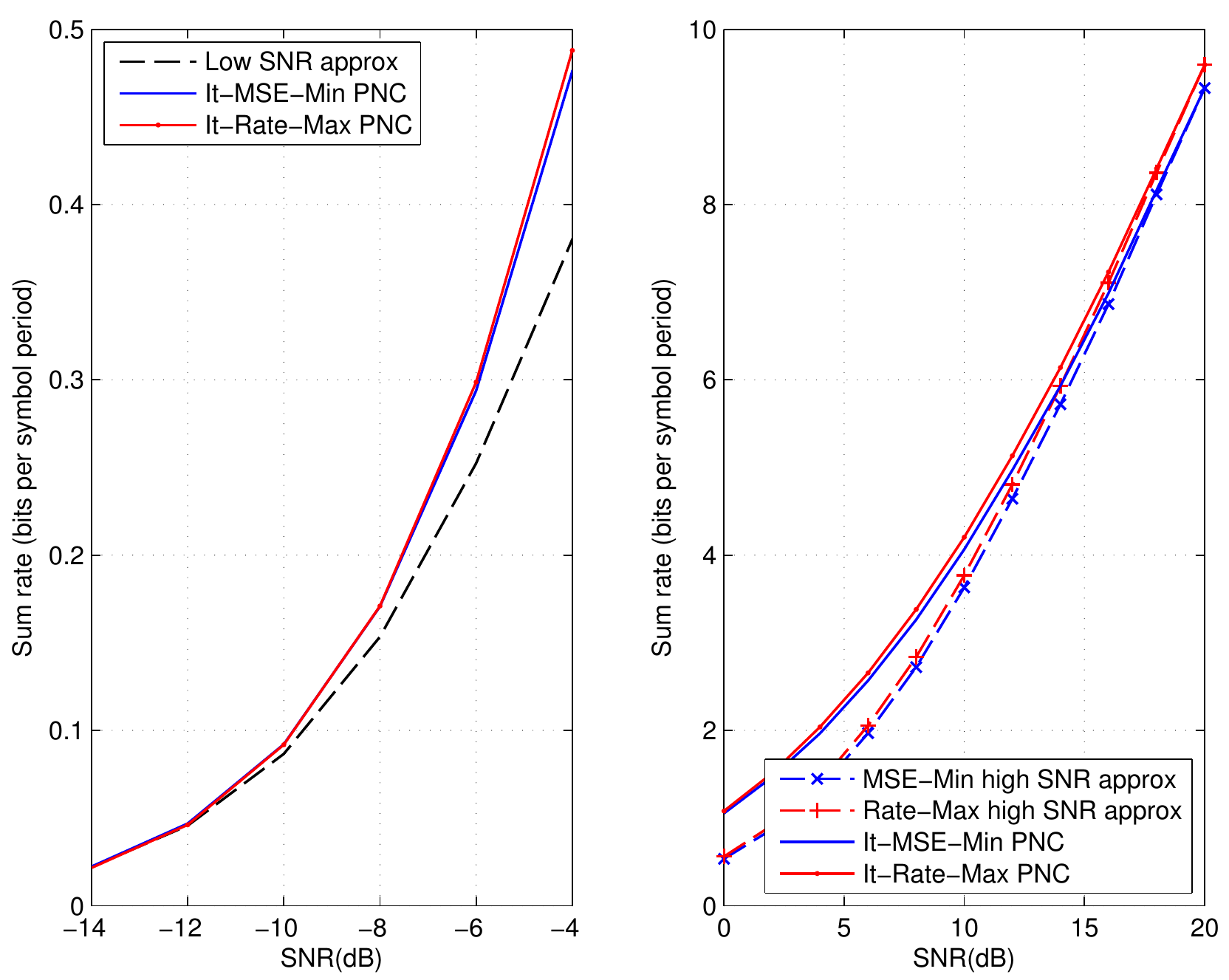}
\caption{The sum rate performance is evaluated for different asymptotically optimal solutions when $N=K=4$. ``Low SNR approx'' is the low-SNR asymptotically optimal sollution given by Theorem $1$; ``MSE-Min high SNR approx'' and ``Rate-Max high SNR approx'' are the high-SNR asymptotically optimal solutions given by Theorem $2$ as well as Proposition $6$ and $7$, respectively. The proposed iterative schemes achieve nearly optimal performance at both low and high SNR. In the medium SNR regime, it shows that the proposed schemes also achieve good performance.}
\label{new}
\end{figure}

\begin{figure}
\centering
\includegraphics[width=6in]{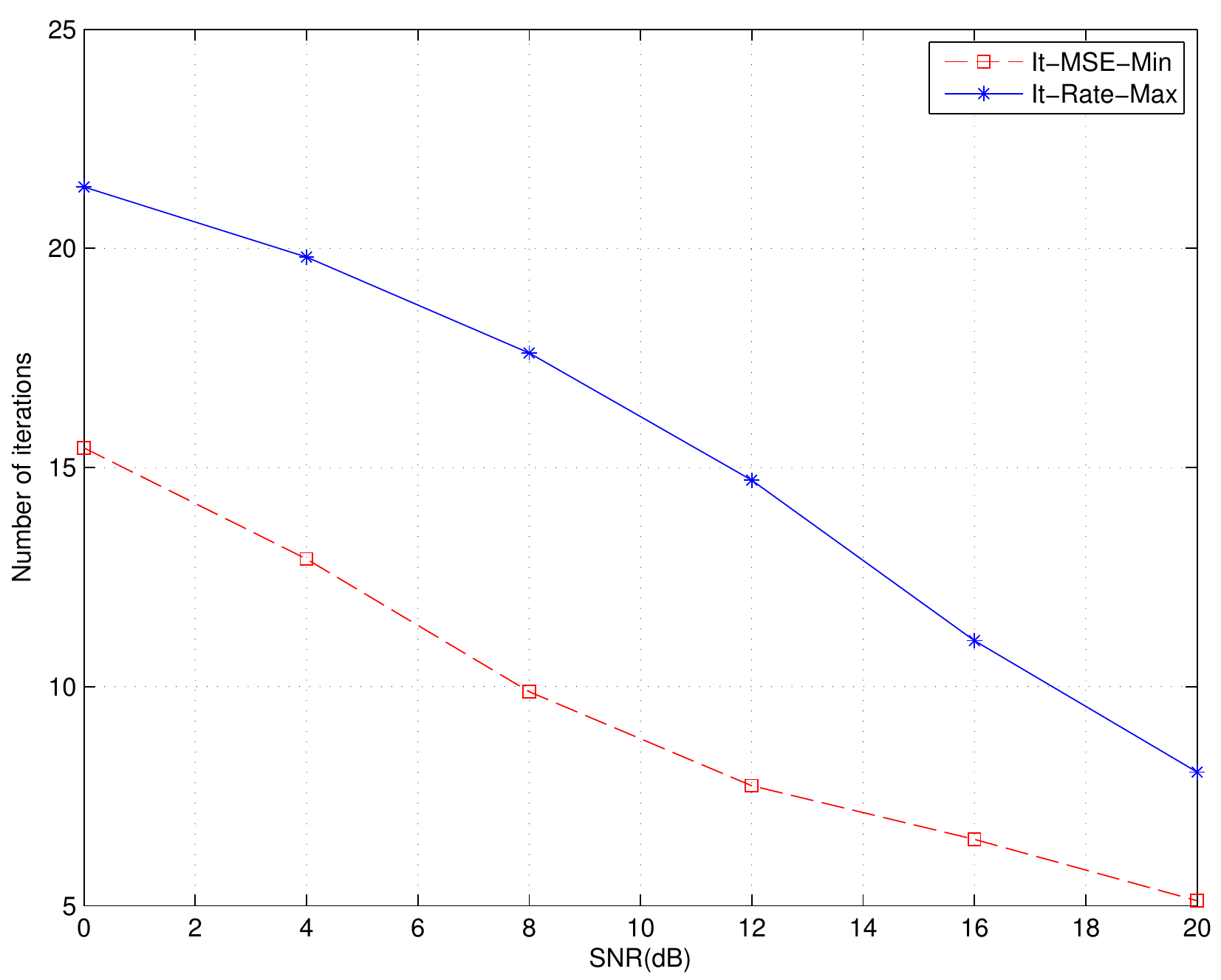}
\caption{The numbers of iterations of the two proposed iterative algorithms
are evaluated for different schemes when $N=K=4$. The average number of iterations decreases as SNR increases. The algorithm for the sum rate optimization problem (i.e., It-Rate-Max) needs more iterations than that for the sum MSE optimization problem (i.e., It-MSE-Min). }
\label{nb_iter}
\end{figure}

\medskip\noindent{\bf\emph{Observation 4}}: {\em With the high-SNR asymptotically optimal solutions as the initials, the proposed iterative algorithms achieve good performance over all SNR regimes.}

In Fig.~\ref{new}, the proposed iterative algorithms are initialized by their high-SNR asymptotically optimal solutions, respectively. It shows that both the It-MSE-Min scheme and the It-Rate-Max scheme asymptotically approaches optimal as SNR decreases and increases. In addition, the high-SNR asymptotically optimal solutions also achieve good performance especially at high SNR, while the low-SNR solution degrades fast as SNR increases. Thus, regarding both the throughput and the computational complexity, the high-SNR asymptotically optimal solutions  are  practical schemes in the engineering view.

In addition, we evaluate the convergence speed of the proposed iterative algorithms.
Our numerical results indicate that the proposed algorithms converge quite fast. For example, in Fig.~\ref{nb_iter}, less than $25$ iterations are required to achieve convergence on average. The average number of iterations decreases as SNR increases. The weighted sum rate optimization problem needs more iterations than the weighted sum MSE optimization problem.

%
%
%

\section{Conclusion}

In this paper, we have proposed a unified approach to iteratively solve the weighted sum MSE minimization and the weighted sum rate maximization problems for the wireless MIMO switching networks with and without PNC. 
We proved that, although the proposed algorithms are suboptimal in general, they can converge to asymptotically optimal solution in the low SNR regime regardless of the initial conditions, and near optimal solution in the high SNR regime with properly setting initial conditions. Numerical results show that the proposed iterative algorithms significantly outperform the existing ZF and MMSE relaying schemes for MIMO switching for all SNR.

This paper makes several assumptions to simplify  the design and analysis of the MIMO switching schemes. For example, we assume that each user has a single antenna, and full channel state information is available to the relay. It is of theoretical and practical interest to investigate the impact of relaxing these assumptions on the MIMO switching design.

\appendices

\section{Proof of Theorem \ref{theorem_low}}   \label{theorem_lowp}

We first consider solving the weighted sum MSE minimization in \eqref{opt_mse}. The global optimal solution of \eqref{opt_mse} should satisfy the KKT conditions, i.e. the results in Proposition \ref{prop_mmse_1} and Proposition \ref{prop_mmse_2}. From Proposition \ref{prop_mmse_1}, the optimal $\al$ can be expressed in terms of $\bar\bG$ as in \eqref{al}. When $\sm^2,\gm^2\rightarrow +\infty$, we obtain
\be \label{al_low}
\al^2= \frac{P_r}{\gm^2 \tr[\bar\bG\bar\bG^H]}+\mathcal{O}(\gm^{-2}).
\ee
Similarly, from Proposition \ref{prop_mmse_2} the optimal $\bB$ and $\bar\bC$ can be expressed in terms of $\bar\bG$ and $\al$ as in \eqref{bc}. As $\sigma^2,\gamma^2 \rightarrow +\infty$, we obtain
\bseq  \label{52}
\be
\bB^{opt} &=& \al^2\sm^{-2}[\bF\bar\bG\bH]_{\text{diag}}[\bP\bQ\bH^H\bar\bG^H\bF^H]_{\rm diag}+\mathcal{O}(\sm^{-2}\gm^{-2}), \label{b_low}\\
\bar\bC^{opt} &=& \al^2\sm^{-2}[\bP\bQ\bH^H\bar\bG^H\bF^H]_{\rm diag}+\mathcal{O}(\sm^{-2}\gm^{-2}). \label{c_low}
\ee
\eseq
Then, the weighted sum MSE \eqref{opt_mse_x1a} is rewritten as
\be
&&\mathcal{J}(\bar\bG,\al,\bB,\bar\bC) \notag \\
&\mathop=\limits^{}&\tr\left[\bW\bP\bQ\bP^T-2\Re\left\{\al^2\sm^{-2} \bW[\bP\bQ\bH^H\bar\bG^H\bF^H]_{\rm diag}\bF\bar\bG\bH\bQ\bP^T \right\}\right.\notag\\
 &&+ \gm^2\al^4\sm^{-4}\bW[\bP\bQ\bH^H\bar\bG^H\bF^H]_{\rm diag}\bF\bar\bG\bar\bG^H\bF^H[\bP\bQ\bH^H\bar\bG^H\bF^H]_{\rm diag}\notag\\
  &&\left.+ \al^{2}\sm^{-2}\bW[\bP\bQ\bH^H\bar\bG^H\bF^H]_{\rm diag}[\bP\bQ\bH^H\bar\bG^H\bF^H]_{\rm diag} \right]+\mathcal{O}(\sm^{-4}\gm^{-4})\label{54a}\\
&\mathop=\limits^{}& \tr[\bW\bP\bQ\bP^T]+\tr\left[-2\Re\left\{\al^2\sm^{-2}\bW [\bP\bQ\bH^H\bar\bG^H\bF^H]_{\rm diag}\bF\bar\bG\bH\bQ\bP^T \right\}\right.\notag\\
  &&\left.+ \al^{2}\sm^{-2}\bW[\bP\bQ\bH^H\bar\bG^H\bF^H]_{\rm diag}[\bF\bar\bG\bH\bQ\bP^T]_{\rm diag} \right]+\mathcal{O}(\sm^{-4}\gm^{-2})\label{54b}\\
&\mathop=\limits^{}&\tr[\bW\bP\bQ\bP^T]- \tr\left[\al^{2}\sm^{-2}\bW[\bP\bQ\bH^H\bar\bG^H\bF^H]_{\rm diag}[\bF\bar\bG\bH\bQ\bP^T]_{\rm diag}\right] \notag\\
&&+\mathcal{O}(\sm^{-4}\gm^{-2})\\
&\mathop=\limits^{}&\tr[\bW\bP\bQ\bP^T]- \tr\left[\al^{2}\sm^{-2}\bW[\bQ\bH^H\bar\bG^H\bF^H\bP]_{\rm diag}[\bP^T\bF\bar\bG\bH\bQ]_{\rm diag}\right] \notag\\
&&+\mathcal{O}(\sm^{-4}\gm^{-2})\\
&\mathop=\limits^{}&\tr[\bW\bP\bQ\bP^T]- \sm^{-2}\sum_{\ell=1}^K q_\ell^2 w_\ell|\bp_\ell^T\bF\bG\bh_\ell|^2 +\mathcal{O}(\sm^{-4}\gm^{-2}), \label{75}
\ee
where \eqref{54a} follows by plugging in \eqref{b_low} and \eqref{c_low}; \eqref{54b} follows from \eqref{al_low}; \eqref{75} follows from \eqref{al_mmse}.

Now we consider the weighted sum rate in \eqref{sum_rate}. In the low SNR regime, this weighted sum rate becomes first order approximated as
\be \label{76}
\sum_{\ell=1}^{K}w_\ell R_\ell&= & \frac{1}{2\sigma^2}\sum_{\ell=1}^{K} q_\ell^2 w_\ell|\bp_\ell^T\bF\bG\bh_\ell|^2+\mathcal{O}(\sm^{-2}) .
\ee
From \eqref{75} and \eqref{76}, we see that minimizing the weighted sum MSE is equivalent to maximizing the weighted sum rate in the low-SNR regime, with the solution given by solving
\begin{subequations} \label{opt_low_converge}
\begin{align}
\mathop\text{maximize}_{\bG}  &\qquad \sum_{\ell=1}^{K} q_\ell^2 w_\ell|\bp_\ell^T\bF\bG\bh_\ell|^2 \label{opt_low_converge_a}\\
\text{subject to} &\qquad  \tr[\bG\bG^H] \leq \frac{P_r}{\gm^2}. \label{opt_low_converge_b}
\end{align}
\end{subequations}
Note that
\be \label{sum_rate_low}
|\bp_\ell^T\bF\bG\bh_\ell|^2&=& \tr[\bp_\ell^T\bF\bG\bh_\ell\bh_\ell^H\bG^H\bF^H\bp_\ell]\\
&=&\tr[\bG^H\bF^H\bp_\ell\bp_\ell^T\bF\bG\bh_\ell\bh_\ell^H]\\
&=&\bg^H vec(\bF^H\bp_\ell\bp_\ell^T\bF\bG\bh_\ell\bh_\ell^H)\\
&\mathop=\limits^{}& \bg^H\left((\bh_\ell\bh_\ell^H)\otimes(\bF^H\bp_\ell\bp_\ell^T\bF)\right)\bg,
\ee
where $\bg=vec(\bG)$, and the last step follows from $vec(\bA\bB\bC)=(\bC^T\otimes\bA)vec(\bB)$ in Lemma $4.3.1$ \cite{roger90}. With the definition in \eqref{psi}, problem \eqref{opt_low_converge} can be equivalently written as
\begin{subequations}  \label{opt_mse_low}
\begin{align}
\mathop \text{maximize}_{\boldsymbol{g}}  &\qquad \bg^H\bPsi\bg\\
\text{subject to} &\qquad  \bg^H\bg\leq\frac{P_r}{\gm^2}. \label{83b}
\end{align}
\end{subequations}
The optimal $\bg$ to the above problem is an eigenvector corresponding to the maximum eigenvalue of $\bPsi$, which concludes the proof.

\section{Proof of Proposition \ref{prop_low_converge}}  \label{prop_low_converge_p}

We first consider the convergence point of Algorithm $1$. Clearly, any convergence point of Algorithm $1$ satisfies both \eqref{10} and \eqref{bc}. In the low SNR regime, \eqref{bc} can be first-order approximated as \eqref{52}. Plugging in  \eqref{52},  \eqref{tGopt_mmse} can be first-order approximated as
\bseq
\be
\bar\bG &\mathop\approx\limits^{}& \gm^{-2}\left(\frac{\sm^2}{P_r}\tr[\bW\bar\bC\bar\bC^H]\right)^{-1}\bF^H \bar\bC^H \bW\bP\bQ\bH^H  \label{59a}\\
&\mathop=\limits^{}&  \frac{\tr[\bar\bG\bar\bG^H]}{\|\bW^{\frac12}\diag\{\bP\bQ\bH^H\bar\bG^H\bF^H\}\|_F^2} \bF^H [\bF\bar\bG\bH\bQ\bP^T]_{\rm diag} \bW\bP\bQ\bH^H \label{59b}\\
&\mathop=\limits^{}&   \xi \bF^H [\bF\bar\bG\bH\bQ\bP^T]_{\rm diag}\bW \bP\bQ\bH^H \label{59c}\\
&\mathop=\limits^{}&   \xi \bF^H\bP [\bP^T\bF\bar\bG\bH\bQ]_{\rm diag} \hat\bW\bQ\bH^H, \label{90}
\ee
\eseq
where $\hat\bW=\bP^T\bW\bP=\diag\{w_1,\cdots,w_K\}$.
In the above, \eqref{59a} follows by noting that as $\sm^2,\gm^2\rightarrow +\infty$,
\bseq
\be
\frac{\sm^2}{P_r} \tr[\bW\bar\bC\bar\bC^H]\bI + \bF^H\bar\bC^H\bW\bar\bC\bF&\approx & \frac{\sm^2}{P_r} \tr[\bar\bC\bar\bC^H]\bI,\\
\bP+\bB&\approx &\bP,\\
\bH\bH^H+\gm^2\bI&\approx &\gm^2\bI.
\ee
\eseq
\eqref{59b} follows by plugging in \eqref{al_low} and \eqref{c_low}. \eqref{59c} follows by defining
\be
\xi=\frac{\tr[\bar\bG\bar\bG^H]}{\|\bW^{\frac12}\diag\{\bP\bQ\bH^H\bar\bG^H\bF^H\}\|^2}.
\ee
\eqref{90} follows from $\bP^T[\bA]_{\rm diag}\bP=[\bP^T\bA\bP]_{\rm diag}$.
Let $\bff_\ell$ to be column $\ell$ of $\bF^H\bP$. Recall that $\bh_\ell$ is column $\ell$ of $\bH$, $\ell=1,\cdots,K$. We obtain
\be
&&\bF^H \bP[\bP^T\bF\bar\bG\bH\bQ]_{\rm diag} \hat\bW\bQ\bH^H \notag \\
&=& \begin{bmatrix}
\bff_1 & \cdots & \bff_K
\end{bmatrix}\diag\{q_1^2w_1\bff_1^H\bar\bG\bh_1,\cdots,q_K^2w_K\bff_K^H\bar\bG\bh_K\}
\begin{bmatrix}
\bh_1 & \cdots & \bh_K
\end{bmatrix}^H \ \ \ \\
&=& \sum_{\ell=1}^K q_\ell^2 w_\ell (\bff_\ell\bff_\ell^H) \bar\bG(\bh_\ell \bh_\ell^H) \\
&=& \sum_{\ell=1}^K q_\ell^2 w_\ell (\bF^H\bp_\ell\bp_\ell^T\bF) \bar\bG(\bh_\ell \bh_\ell^H)\\
&\mathop=\limits^{}& \bPsi\bar\bg. \label{94}
\ee
where the last step follows from $vec(\bA\bB\bC)=(\bC^T\otimes\bA)vec(\bB)$ and the definition in \eqref{psi}. With \eqref{94}, the optimal precoder in \eqref{90} is rewritten as
\be  \label{61}
\bar\bg&=&\xi\bPsi\bar\bg,
\ee
implying that the convergence point of Algorithm $1$ must be one of the eigenvectors of $\bPsi$. Furthermore, the iterative process of Algorithm $1$ is equivalent to recursively calculating $\bPsi\bar\bg$ and then updating $\bar\bg$ with $\bPsi\bar\bg$ normalized by $\xi$.
This process is called {\em power iteration} in linear algebra \cite{matrix}, with the fixed point given by the eigenvector corresponding to the largest eigenvalue of $\bPsi$. Thus, Algorithm $1$ converges to the optimal solution in \eqref{44}.

Next, we consider the convergence point of Algorithm $2$. We aim to show that this convergence point is also given by \eqref{61}. To this end, we note that the similarity between the weighted sum rate maximization problem and the weighted sum MSE minimization problem
is established based on \eqref{wbc} as pointed out in Section IV.B. We also note that the convergence point of the weighted sum MSE minimization problem is given by \eqref{61}. Therefore, it suffices to show that, in the low SNR regime, the weighted sum rate maximization is equivalent to the weighted sum MSE minimization. From \eqref{wbc}, we thus need to show that $\bW=\bT\bSm^{-1}\approx \bT$, i.e., $\bSm\approx \bI$, in the low SNR regime, as detailed below.

When $\sm^2,\gm^2\rightarrow +\infty$, the scaling factor $\al$ can be approximated as in \eqref{al_low}, and hence the covariance \eqref{crr} can be approximated as
\be  \label{96}
\mathcal{C}_{z_iz_i}\approx\sm^2.
\ee
Together with \eqref{crx} and \eqref{al_low}, we have
\be
\Sm_i &\approx& 1- \al^2 (\bp_i^T\bF\bar\bG\bh_i)^* \sm^{-2} \al^2 \bp_i^T\bF\bar\bG\bh_i\approx 1,\quad i=1,\cdots,K.
\ee
or equivalently,
\be
\bSm &\approx& \bI, \label{100x}
\ee
which completes the proof.

\section{Proof of Theorem \ref{theorem_high}}  \label{theorem_highp}

Recall that $\bF$ is a $K$-by-$N$ matrix. Then
\be
\left(\bF^H\bar\bC^H\bW\bar\bC\bF\right)^{\dag} = \bF^{\dag}\bar\bC^{-1} \bW^{-1}(\bar\bC^H)^{-1} (\bF^H)^{\dag}.
\ee
From Proposition \ref{prop_mmse_1}, the optimal $\bar\bG$ for problem \eqref{opt_mse} can be expressed in terms of $\bB$ and $\bar\bC$ as in \eqref{tGopt_mmse}. When $\sm^2,\gm^2\rightarrow 0$, we obtain
\be  \label{68x}
\bG^{opt} &=& \al \bar\bG^{opt}\\
&=&\al \bF^{\dag} \bar\bC^{-1} (\bP+\bB) \bH^{\dag}\\
&=& \bF^{\dag} \bC^{-1} (\bP+\bB) \bH^{\dag},
\ee
which completes the proof for the weighted sum MSE minimization part.

Now we consider the weighted sum rate maximization. From \eqref{Gnopt_rate} and the discussion therein, the
optimal precoder for the weighted sum rate maximization can be expressed in terms of $\bar\bDt$ and $\bSm$ as in \eqref{g2}. Letting $\sm^2,\gm^2\rightarrow 0$ in \eqref{Gnopt_rate}, we obtain
\be
\bG^{opt} = \al^{} \bF^{\dag}\bar\bOm^{-1}(\bP+\bar\bOm\bar\bDt)\bH^{\dag}.
\ee
Recall that both $\bar\bDt$ and $\bSm$ are diagonal. Thus, the theorem is proved by variable substitutions of letting $\bB=\bar\bOm\bar\bDt$ and $\bC=\al^{-1}\bar\bOm$.

\section{Proof of Proposition \ref{prop_high_converge}} \label{prop_high_converge_p}

We first see the high-SNR approximation of the solution \eqref{tGopt_mmse}, \eqref{b} and \eqref{c} for the weighted sum MSE minimization problem. When $\sm^2,\gm^2\rightarrow 0$, from \eqref{tGopt_mmse}, the optimal solution of $\bG$ can be rewritten as
\be \label{xG}
\bG^{opt}=\bF^{\dag}\bC^{-1}(\bP+\bB)\bH^{\dag}.
\ee
Plugging \eqref{xG} into \eqref{b} and \eqref{c}, we have
\be
\bB^{opt}=\bB\quad \text{and}\quad \bC^{opt}=\bC.
\ee
This result indicates that the solution of Algorithm $1$ is trapped at the initial value at high SNR.

Next, we see the high-SNR approximation of the solution \eqref{Gnopt_rate} and \eqref{Smi}. Plugging \eqref{asymp_high} into  \eqref{Smi}, we have
\be
\om_i^{opt}=\om_i\quad \text{and}\quad \Sm_i^{opt}\rightarrow 0,\quad i=1,\cdots,K.
\ee
Note that the term $\frac{\sm^2}{P_r}\tr[\bT\bSm^{-1}\bar\bOm \bar\bOm^H]$ in \eqref{Gnopt_rate} is still finite due to $\sm^2\rightarrow 0$. This result indicates that the solution of Algorithm $2$ is also trapped at the initial condition. Thus, the proof is completed.



\section{Proof of Proposition \ref{prop_high_mmse_2}}  \label{prop_high_mmse_2p}

We start with the proof of part (i). Problem \eqref{pmse} can be equivalently expressed as
\begin{subequations} \label{tpmse}
\begin{align}
\mathop\text{minimize}_{\mathclap{{b_1,\cdots,b_K},{c_1,\cdots,c_K}}}  &\qquad \sum_{\imath=1}^{K} w_\imath \left(\gm^2h_{\jmath\jmath}|b_\jmath|^2+2\gm^2\Re\{h_{\imath\jmath}^*b_\jmath\}+\sm^2|c_\jmath|^2+\gm^2 h_{\imath\imath}\right)  \label{112a}\\
\text{subject to} &\qquad  \sum_{\imath=1}^{K} \left( q_{\pi^{-1}(\imath)}f_{\imath\imath}|c_\imath|^{-2}+q_{\imath}f_{\imath\imath}|c_\imath|^{-2} |b_\imath|^2+2\Re\{q_{\imath}f_{\jmath\imath}(c_\jmath^*)^{-1}c^{-1}_\imath b_\imath\}\right)\le P_r. \label{112b}
\end{align}
\end{subequations}
For fixed $\{c_\jmath\}$, both the objective function \eqref{112a} and the constraint \eqref{112b} are quadratic functions of $\{b_\jmath\}$. Hence it is straightforward to see that problem \eqref{tpmse} (or equivalently \eqref{pmse}) is convex in $\{b_\jmath\}$. By the Lagrangian method, The optimal $\{b_\jmath\}$ can be  obtained as
\be
b_\jmath = -\frac{\gm^2 w_\jmath h_{\imath\jmath} + \lb q_{\jmath} f_{\jmath,\pi(\jmath)}(c_\jmath^*)^{-1}c^{-1}_{\pi(\jmath)}}{\gm^2w_{\jmath} h_{\jmath\jmath} + \lb q_{\jmath} f_{\jmath\jmath}|c_\jmath|^{-2} },\quad \jmath = 1,\cdots,K,
\ee
where $\lb$ is a scalar to meet the relay power constraint.

Next, we prove part (ii) and (iii). Since the weighted sum MSE \eqref{112a} is independent of the phases of $\{c_\jmath\}$, the optimal phases must minimize the relay power consumption in \eqref{112b}. Note that \eqref{112b} can be rewritten as
\bal  \label{100}
\sum_{\imath=1}^{K}\left(q_{\imath}|c^{-1}_\jmath+f_{\jmath\imath}b_\imath c^{-1}_\imath|^2 +q_{\pi^{-1}(\imath)}f_{\imath\imath}|c_\imath|^{-2} + q_{\imath}
(f_{\imath\imath}|b_\imath|^2 |c_\imath|^{-2}-|f_{\jmath\imath}b_\imath|^2 |c_\imath|^{-2}-|c_\jmath|^{-2})\right)\le P_r.
\end{align}
Therefore, for any given $\{b_\imath\}$, $c_\jmath^{-1}$ and $f_{\jmath\imath}b_\imath c^{-1}_\imath$ have opposite phases, or equivalently, the optimal phases of $\{c_\jmath\}$ satisfy
\bal   \label{phase_c}
\angle c_\jmath=&\angle c_\imath-\angle f_{\jmath\imath}-\angle b_{\imath}-\pi,\quad \jmath=1,\cdots,K.
\end{align}
With \eqref{phase_c}, the relay power consumption in \eqref{100} reduces to
\be  \label{82}
\sum_{\imath=1}^{K} \left(q_{\pi^{-1}(\imath)} f_{\imath\imath}|c_\imath|^{-2}+q_{\imath} f_{\imath\imath}|c_\imath|^{-2} |b_\imath|^2-2 q_{\imath}|f_{\jmath\imath}||c_\jmath|^{-1}|c_\imath|^{-1} |b_\imath|\right).
\ee
Let $\eta_\imath=|c_\imath|^{-2}$, $\imath=1,\cdots,K$. Clearly, the objective \eqref{112a} is convex in $\{\eta_\imath\}$. In addition, $|c_{\jmath}|^{-1}|c_\imath|^{-1}=\sqrt{\eta_\jmath \eta_\imath}$, the geometric mean of $\eta_\jmath$ and $\eta_\imath$, is concave in $\eta_\jmath$ and $\eta_\imath$ \cite{boyd}. As non-negative weighted sums preserve concavity, $\sum_{\imath=1}^{K}\sqrt{\eta_\jmath \eta_\imath}$ is concave in $\{\eta_\imath\}$, and hence the power in \eqref{82} is convex in $\{\eta_\imath\}$. Therefore, the problem \eqref{tpmse} is a convex problem of $\{\eta_\imath\}$, which concludes the proof.



\section{Proof of Proposition \ref{prop_high_rate_2}}  \label{prop_high_rate_2p}

The proof of Proposition \ref{prop_high_rate_2} is similar to that of Proposition \ref{prop_high_mmse_2}. The sketch of the proof is given as follows.
We start with the proof of part (i). For fixed $\{c_\jmath\}$, we  use the Lagrangian method to solve the problem \eqref{pmse2}. The Lagrangian function is written as
\bal  \label{85}
\mathcal{L}(\{b_\jmath\},\lb)=&\frac{1}{2}\sum_{\jmath=1}^{K}  \log \frac{t_j q_j|c_\jmath|^{-2}}{\gamma^2 |c_\jmath|^{-2}\left[ h_{\jmath\jmath} |b_\jmath|^2+2\Re\{h_{\imath\jmath}^*b_\jmath\}+ h_{\imath\imath}\right]
+\sm^2} \notag\\
& + \lb\left(\sum_{\imath=1}^{K} (q_{\pi^{-1}(\imath)}f_{\imath\imath}|c_\imath|^{-2}+q_{\imath}f_{\imath\imath}|c_\imath|^{-2} |b_\imath|^2+2\Re\{q_{\imath}f_{\jmath\imath}(c_\jmath^*)^{-1}c^{-1}_\imath b_\imath\})- P_r\right)\ \
\end{align}
Setting the derivatives w.r.t. $\{b_\imath\}$ to zero, we obtain
\bal
(\gamma^2  h_{\i\i} |b_\i|^2+2\gamma^2\Re\{h_{\pi^{-1}(\imath),\i}^*b_\i\}&+ \gamma^2h_{\pi^{-1}(\imath),\pi^{-1}(\imath)}+\sm^2 |c_\jmath|^{2})  \notag \\
\times (\gm^2 h_{\imath\imath}b_{\imath}+\gm^2 h_{\pi^{-1}(\imath),\imath}) 
& + \lb q_{\imath}\left(f_{\imath\imath}|c_\imath|^{-2} b_\imath +f_{\imath\jmath}(c_\imath^*)^{-1}c^{-1}_\jmath \right)=0.
\end{align}
The above is a cubic equation w.r.t. $b_\jmath$, which allows closed-form solution \cite{cubic}.

We next prove part (ii) and (iii). Since the sum rate in \eqref{60a} is independent of the phases of $\{c_\jmath\}$ and problem \eqref{pmse2} has the same constraint as  problem \eqref{tpmse}, we conclude that the optimal phases of $\{c_\jmath\}$ also satisfy \eqref{phase_c}. Moreover, let $\eta_\imath=|c_\imath|^{-2}$, $\imath=1,\cdots,K$. It can be shown that the constraint \eqref{60b} is convex in $\{\eta_\imath\}$, and that the objective \eqref{60a} is concave in $\{\eta_\imath\}$. Thus, problem \eqref{pmse2} is a convex problem in $\{\eta_\imath\}$, which completes the proof.

\section{Results for Non-PNC Schemes} \label{prop_non_pnc}

In this case, we have $\bB={\bf0}$. Then the optimal precoders are given in the following propositions with the proof, respectively.

\begin{Proposition}  \label{prop_high_mmse_1}
For the non-PNC case, i.e., $\bB={\bf0}$, the optimal solution to \eqref{pmse} is given by $\bC^\infty$ with the $i$th diagonal element being
\bal
c_j^\infty=&\left(\frac{\sqrt{q_{i}  f_{jj}}\sum_{\ell=1}^K \sqrt{w_\ell q_{\pi^{-1}(\ell)} f_{\ell\ell} }}{Pr\sqrt{w_j}}\right)^{\frac{1}{2}},\quad j = 1,\cdots,K. \label{cin}
\end{align}
\end{Proposition}

\begin{proof}
Plugging in \eqref{asymp_high}, the weighted sum MSE in \eqref{opt_mse_x1a} can be rewritten as
\be
\mathcal{J} = \tr[\bW(\gm^2\bP(\bH^H\bH)^{-1}\bP^T +\sm^2\bC\bC^H)],
\ee
and the relay power consumption in \eqref{opt_mse_1b} is rewritten as
\be
\text{Tr}[\bG\left(\bH\bQ\bH^H+\gm^2\bI\right)\bG^H]&\approx& \text{Tr}[\bG\bH\bQ\bH^H\bG^H]\\
&=&\tr[(\bF\bF^H)^{-1}\bC\bP\bQ\bP^T\bC^H].
\ee
Then, Problem \eqref{opt_mse} can be minimized equivalently by solving
\begin{subequations} \label{hmse}
\begin{align}
\mathop\text{minimize}_{c_1,\cdots,c_K}  &\qquad \sum_{\ell=1}^K w_\ell|c_\ell|^2 \\
\text{subject to} &\qquad  \sum_{\ell=1}^K \frac{f_{\ell\ell}Q_{\pi^{-1}(\ell),\pi^{-1}(\ell)}}{|c_\ell|^2} \leq P_r,
\end{align}
\end{subequations}
where $c_\ell$ is the $\ell$th diagonal element of $\bC$. It can be readily verified that problem \eqref{hmse} is convex. Solving the KKT conditions \cite{boyd}, we obtain the solution in \eqref{cin}.
\end{proof}

\begin{Proposition}  \label{prop_high_rate_1}
For $\bB={\bf0}$, the optimal solution to \eqref{pmse2} is written as
\be \label{alpha_j}
c_\jmath= \frac{\gm}{\sm}(2h_{\i\i})^{\frac{1}{2}}\left(\left( 1+\frac{2\gm^2h_{\i\i}}{\lb\sm^2f_{\j\j}q_\j}\right)^{\frac{1}{2}}-1\right)^{-\frac{1}{2}}, \quad \jmath=1,\cdots,K,
\ee
where $\lambda$ is a scaling factor to guarantee that the relay transmits with its maximum power.
\end{Proposition}

\begin{proof}
When $\sm^2\rightarrow 0$, the weighted sum rate \eqref{sum_rate} is approximated as
\bal \label{sum_rate_high2}
\mathcal{R}(\{c_\jmath\})\approx &\frac{1}{2}\sum_{\imath=1}^{K} \log \frac{t_j q_j|\bp_\imath^T\bF\bG\bh_\imath|^2} {\gamma^2|\bp_\imath^T\bF\bG|^2 +\sigma^2} \\
=&\frac{1}{2}\sum_{\jmath=1}^{K}  \log \frac{t_j q_j|c_{\jmath}|^{-2}}{\gamma^2h_{\imath\imath}|c_{\jmath}|^{-2}+\sigma^2},
\end{align}
The constraint of the relay power consumption \eqref{opt_rate_b} is approximated as
\bal
\tr[\bG\bH\bQ\bH^H\bG^H]=\sum_{\jmath=1}^K f_{\jmath\jmath} q_\jmath|c_\jmath|^{-2} \leq P_r, \label{relay_power_2}
\end{align}
We obtain the solution in \eqref{alpha_j} by solveing the KKT conditions.
\end{proof}

\bibliographystyle{IEEEtran}
\bibliography{MIMO_switch}

\end{document}